\documentclass[11pt,twoside,a4paper,notitlepage]{article}

\def\DONOTINSERTCOMMENTS{}

\usepackage{array}
\newcolumntype{x}[1]{>{\centering\arraybackslash\hspace{0pt}}p{#1}}

\usepackage{thm-restate}

\usepackage{ifthen}

  \usepackage[a4paper,margin=2.5cm]{geometry}
  \usepackage{tgtermes}
  \usepackage[scale=.92]{tgheros}
  \usepackage{tgcursor}

\usepackage{microtype}

  \usepackage[english]{babel}

  \usepackage{fancyhdr}

\usepackage{subcaption}

\usepackage{bussproofs}

\usepackage{dsfont}
\usepackage{stmaryrd}

\usepackage{ifthen}

\usepackage{ifthen}



\usepackage[utf8]{inputenc}

\usepackage{amsmath}
\usepackage{amssymb}
\usepackage{amsfonts}

\usepackage{mathtools}

\usepackage{bussproofs}

\usepackage{bm}  

\usepackage{subcaption}

\usepackage{varioref}

\usepackage{xspace}

  \usepackage{amsthm}
  \usepackage{amsthm-boldfont}

\usepackage{jncaptionsheadings}

\usepackage{ifthen}
\usepackage{xspace}
\usepackage{varioref}

\DeclareMathAlphabet{\mathsfsl}{OT1}{cmss}{m}{sl}

\newcommand{\bigoh}[1]{\mathrm{O} ( #1 )}

\newcommand{\bigomega}[1]{\Omega ( #1 )}

\DeclareMathOperator{\polylog}{polylog}

\providecommand{\abs}[1]{\lvert#1\rvert}

\providecommand{\ABS}[1]{\left\lvert#1\right\rvert}

\newcommand{\ceiling}[1]{\lceil #1 \rceil}

\newcommand{\set}[1]{\{ #1 \}}

\newcommand{\setdescr}[3][\mid]{\set{ #2 #1 #3 }}

\newcommand{\setsize}[1]{\lvert#1\rvert}

\newcommand{\intersection}{\cap}

\newcommand{\union}{\cup}
\newcommand{\Union}{\bigcup}

\newcommand{\Lor}{\bigvee}
\newcommand{\Land}{\bigwedge}

\newcommand{\limpl}{\ensuremath{\rightarrow}}

\newcommand{\olnot}[1]{\overline{#1}}

\newcommand{\clwidth}{k}

\newcommand{\xcnf}[1]{\mbox{\ensuremath{#1}-CNF}\xspace}
\newcommand{\kcnf}{\xcnf{\clwidth}}

\newcommand{\complclassformat}[1]{\textrm{\upshape{\textsf{#1}}}\xspace}

\newcommand{\Pclass}{\complclassformat{P}}

\newcommand{\NCclass}{\complclassformat{NC}}

\newcommand{\NC}{\complclassformat{NC}}

\newcommand{\introduceterm}[1]{{\emph{#1}}}

\newcommand{\eqperiod}{\enspace .}

\newcommand{\eg}{for instance\xspace} 
\newcommand{\ie}{i.e.,\ }

\newcommand{\etal}{et al.\@\xspace}

\newcommand{\refsec}[1]{Section~\ref{#1}}

\newcommand{\refapp}[1]{Appendix~\ref{#1}}

\newcommand{\refth}[1]{Theorem~\ref{#1}}

\newcommand{\refthm}[1]{Theorem~\ref{#1}}

\newcommand{\reflem}[1]{Lemma~\ref{#1}}

\newcommand{\refobs}[1]{Observation~\ref{#1}}

\usepackage[colorlinks=true,citecolor=blue]{hyperref}

\newtheorem{standardlocalcounter}{Dummy}[section]

\theoremstyle{plain}    

\newtheorem{theorem}[standardlocalcounter]{Theorem}
\newtheorem{lemma}[standardlocalcounter]{Lemma}
\newtheorem{proposition}[standardlocalcounter]{Proposition}
\newtheorem{corollary}[standardlocalcounter]{Corollary}
\newtheorem{observation}[standardlocalcounter]{Observation}

\theoremstyle{definition}

\newtheorem{definition}[standardlocalcounter]{Definition}
\newtheorem{claim}[standardlocalcounter]{Claim}

\theoremstyle{remark}

\usepackage{ifpdf}

\usepackage{graphicx}

\ifpdf         
\fi

\newcount\hours
\newcount\minutes
\def\SetTime{\hours=\time
\global\divide\hours by 60
\minutes=\hours
\multiply\minutes by 60
\advance\minutes by-\time
\global\multiply\minutes by-1 }
\SetTime
\def\now{\number\hours:\ifnum\minutes<10 0\fi\number\minutes}

\newcommand{\formf}{\ensuremath{F}}

\newcommand{\clc}{\ensuremath{C}}
\newcommand{\cld}{\ensuremath{D}}

\newcommand{\clausesetformat}[1]{\ensuremath{\mathbb{#1}}}
\newcommand{\clsc}{\clausesetformat{C}}

\newcommand{\setsofvarsorlit}[2]{\mathit{#1}({#2})}
\newcommand{\Setsofvarsorlit}[2]{\mathit{#1}\bigl({#2}\bigr)}

\newcommand{\vars}[1]{\setsofvarsorlit{Vars}{#1}}

\newcommand{\Vars}[1]{\Setsofvarsorlit{Vars}{#1}}

\newcommand{\restrict}[2]{{{#1}\!\!\upharpoonright_{#2}}}

\newcommand{\lengthstd}{L}

\newcommand{\spacestd}{s}

\newcommand{\ib}[1]{\llbracket{#1}\rrbracket}
\renewcommand{\bigoh}[1]{O(#1)}
\newcommand{\bigohtilde}[1]{\tilde{O}({#1})}
\newcommand{\bigomegatilde}[1]{\tilde\Omega({#1})}
\renewcommand{\olnot}[1]{\neg #1}

\newcommand{\cpstar}{CP$^*$\xspace}

\newcommand{\cc}{\ensuremath{\mathsf{cc}}}
\newcommand{\rank}{\ensuremath{\mathrm{rank}}}
\newcommand{\NS}{\ensuremath{\mathsf{NS}}}
\newcommand{\gap}{\ensuremath{\mathrm{gap}}}
\newcommand{\Pcc}{\mathsf{P}^{\cc}}
\newcommand{\Search}{\ensuremath{\mathsf{Search}}}
\newcommand{\Cert}{\ensuremath{\mathrm{Cert}}}
\newcommand{\rpeb}{\ensuremath{\mathrm{rpeb}}}
\newcommand{\speb}{\mathrm{speb}}
\newcommand{\static}{\mathrm{static}}
\newcommand{\pred}{\mathrm{pred}}
\newcommand{\desc}{\mathrm{desc}}
\newcommand{\head}{\mathrm{head}}

\newcommand{\depth}[1]{\mathrm{depth}(#1)}
\newcommand{\mdepth}[1][f]{\mathrm{mD}(#1)}
\newcommand{\mkw}[1][f]{\mathrm{mKW}(#1)}

\renewcommand{\formf}{\mathcal{C}}

\newcommand{\equalgadget}{\ensuremath{\mathrm{EQ}}}

\newcommand{\coefficientbits}{c}
\newcommand{\numvertices}{n}
\newcommand{\numvariables}{n}
\newcommand{\liftedformula}[2]{{#1} \circ {#2}}
\newcommand{\liftedvar}[3]{{#2}_{{#1}_{#3}}} \newcommand{\alicevar}[2]{\liftedvar{#1}{x}{#2}}
\newcommand{\bobvar}[2]{\liftedvar{#1}{y}{#2}}

\newcommand{\vargeneric}{v}
\newcommand{\varcoefficient}{a_j}
\newcommand{\freecoefficient}{c}
\newcommand{\lingadgetname}{L}
\newcommand{\lingadget}[1]{\lingadgetname(#1)}
\newcommand{\gadgetname}{g}
\newcommand{\boolgadget}[1]{\gadgetname(#1)}
\newcommand{\gadgetarity}{q}
\newcommand{\encodingtime}{\ell_\gadgetname}
\newcommand{\encodingspace}{s_\gadgetname}
\newcommand{\safetyconstant}{K}
\newcommand{\vertexseq}{W}

\newcommand{\prednode}[1]{\mathrm{pred}(#1)}
\newcommand{\predvertex}{u}

\newcommand{\succvertex}{w}
\newcommand{\trimctant}{i}
\newcommand{\ineqsize}{m}

\newcommand{\xnor}{\mathrm{XNOR}}
\newcommand{\eq}{\equalgadget}
\newcommand{\eqloglog}{\eq_{\log\log\numvertices}}

\newcommand{\extrainequalities}{\xi}

\newcommand{\ctaxiom}[1][\alpha]{\clc_{#1}}
\newcommand{\ctaccumulator}[1][\binarylimit]{T_{#1}}

\newcommand{\makelit}[2]{\ell(#1,#2)}
\newcommand{\polarity}{b}
\newcommand{\binarystring}{b}
\newcommand{\binarylimit}{B}
\newcommand{\axiom}[1]{I_{#1}}
\newcommand{\axiompremise}{J_\binarystring}
\newcommand{\axiomconclusion}{\lingadget{\succvertex}}
\newcommand{\slackterm}{S}

\newcommand{\lingadgetpart}[1]{\lingadgetname^+_{#1}}
\newcommand{\boolgadgetpart}[1]{\gadgetname^+_{#1}}

\newcommand{\pebblingformula}[1][G]{\mathrm{Peb}_{#1}}
\newcommand{\decisiontree}{\mathrm{DT}}
\newcommand{\paritydecisiontree}{\mathrm{PDT}}

\newcommand{\pebbling}{\mathcal{P}}
\newcommand{\reversepebbling}{R}

\usepackage[usenames,dvipsnames]{color}
\usepackage{numname}
\newcounter{authorcount}
\newcommand{\newauthor}[3]{\newcounter{#1comment}
\expandafter\newcommand\csname #1comment\endcsname[1]{\ifdefined\DONOTINSERTCOMMENTS\relax\else\medskip\par\noindent
{\bfseries \scshape \footnotesize #2's comment
\stepcounter{#1comment}\csname the#1comment\endcsname}:
{\sffamily \itshape \scriptsize\textcolor{#3}{##1}\par}
\medskip\fi}
\stepcounter{authorcount}
\expandafter\edef\csname #1ordinal\endcsname{\theauthorcount}
\expandafter\newcommand\csname theauthor#1\endcsname{the \ordinaltoname{\csname #1ordinal\endcsname} author\xspace}
\expandafter\newcommand\csname Theauthor#1\endcsname{The \ordinaltoname{\csname #1ordinal\endcsname} author\xspace}
}

\newauthor{sfdr}{Susanna}{OliveGreen}
\newauthor{om}{Or}{TealBlue}
\newauthor{jn}{Jakob}{Blue}
\newauthor{tp}{Toni}{Red}
\newauthor{rr}{Robert}{Orchid}
\newauthor{mv}{Marc}{Orange}

  \numberwithin{equation}{section}

\begin{document}

\newgeometry{margin=1.0in,top=1.7in,bottom=1in}

\begin{center}
  {\LARGE
    Lifting with Simple Gadgets and \\ Applications to Circuit and Proof Complexity
}
\\[1cm] \large

\setlength\tabcolsep{2em}
\begin{tabular}{x{4.1cm}cx{4.1cm}}
Susanna F. de Rezende &
Or Meir &
Jakob Nordström 
\\[-.7mm]
\small\slshape 
Mathematical Institute of~the 
&
\small\slshape University of Haifa &                                     
\small\slshape University of Copenhagen
\\[-.7mm]
\small\slshape Czech Academy of Sciences
& & 
\small\slshape  \mbox{KTH Royal Institute} of~Technology	
 \\
  \rule{0pt}{1ex} & & \\
Toniann Pitassi &
                  Robert Robere &
                                  Marc Vinyals
  \\[-.7mm]
\small\slshape University of Toronto  &
\small\slshape DIMACS &
                                  \small\slshape Technion  \\[-.7mm]
\small\slshape Institute for Advanced Study & \small\slshape Institute for Advanced Study  &
\end{tabular}

  \thispagestyle{empty}

  \pagestyle{fancy}
  \fancyhead{}
  \fancyfoot{}

  \fancyhead[CE]{\slshape
    NULLSTELLENSATZ AND POLYNOMIALLY BOUNDED COEFFICIENTS IN CUTTING PLANES}
  \fancyhead[CO]{\slshape \nouppercase{\leftmark}}
  \fancyfoot[C]{\thepage}

  \setlength{\headheight}{13.6pt}

\begin{abstract}
We significantly strengthen and generalize the theorem lifting
  Nullstellensatz degree to monotone span program size by Pitassi and
  Robere (2018) so that it works for
any
  gadget with high
  enough rank, in particular, for useful gadgets such as
  \emph{equality} and \emph{greater-than}. We apply our generalized
  theorem to solve two open problems:
  \begin{itemize}
    
  \item 
We present the first result that demonstrates a separation in proof
    power for cutting planes with unbounded versus polynomially bounded
    coefficients.
    Specifically, we exhibit CNF formulas that can be refuted in
    quadratic length and constant line space in cutting planes with
    unbounded coefficients, but for which there are no refutations in
    subexponential length and subpolynomial line space if coefficients
    are restricted to be of polynomial magnitude.

  \item 
    We give the first 
explicit
    separation between monotone
    Boolean formulas and monotone real formulas. Specifically, we give
    an explicit family of functions that can be computed with monotone
    real formulas of nearly linear size but require monotone Boolean
    formulas of exponential size. Previously only a non-explicit
    separation was known.  
    
  \end{itemize}
  
  An important technical ingredient, which may be of independent
  interest, is that we show that the Nullstellensatz degree of refuting
  the pebbling formula over a DAG~$G$ over any field coincides exactly
  with the reversible pebbling price of~$G$. 
In particular, this implies
  that the standard decision tree complexity and the parity
  decision tree complexity of the corresponding falsified clause search
  problem are equal.

  \end{abstract}

 \end{center}

\newpage

\restoregeometry
\pagenumbering{arabic}

\section{Introduction}
\label{sec:intro}

{\emph{Lifting theorems}} in complexity theory are a method of
transferring lower bounds
in a weak computational model into lower bounds for
a more powerful computational model, via function composition.
There has been an explosion of lifting theorems in the last
ten years, essentially reducing communication lower bounds
to query complexity lower bounds.

Early papers that establish lifting theorems include
Raz and McKenzie's
separation of the monotone $\NCclass$ hierarchy
~\cite{RM99Separation} (by lifting decision tree complexity
to deterministic communication complexity), and
Sherstov's pattern matrix method
~\cite{Sherstov11Pattern} which lifts
(approximate) polynomial degree to (approximate) matrix rank.
Recent work has established query-to-communication
lifting theorems in a variety of models,
leading to the resolution of many
longstanding open problems in many areas of computer science.
Some examples include the resolution of open questions in communication complexity~\cite{GPW15Deterministic,GLMWZ15Rectangles,GKPW17QueryToCommunicationPtoNP,GJPW17Randomized,GPW18Landscape},
monotone complexity~\cite{RPRC16Exponential,PR17StronglyExponential,PR18LiftingNS},
proof complexity~\cite{HN12VirtueSuccinctProofs,GP18CommunicationLowerBounds,dRNV16LimitedInteraction,GGKS18MonotoneCircuit},
extension complexity of linear and semidefinite programs~\cite{KMR17ApproximatingRectangles,GJW18ExtensionComplexity,LRS15LowerBoundsSDP},
data structures~\cite{CKLM18Simulation}
and finite model theory~\cite{BN16QuantifierDepth}.

Lifting theorems have the following form: given functions $f\colon \{0,1\}^n \to \{0,1\}$ (the ``outer function'')  and $g\colon\mathcal{X} \times \mathcal{Y} \to \{0,1\}$ (the ``gadget''), a lower bound for $f$ in a weak computational model implies a lower bound on $f \circ g^n$ in a stronger computational model.
The most desirable lifting theorems are the most general ones.
First, it should hold for \emph{any} outer function, and ideally $f$ should be allowed to be a partial function or a relation (i.e., a search problem).
Indeed, nearly all of the applications mentioned above require lifting where the outer function is a relation or a
partial function.
Secondly, it is often desirable that the gadget is as small as possible.
The most general lifting theorems established so far, for example lifting
theorems for deterministic and randomized communication complexity, require at least logarithmically-sized gadgets;
if these theorems could be improved generically to hold for constant-sized gadgets then many of the current theorems would be vastly improved.
Some notable examples where constant-sized gadgets are possible include Sherstov's degree-to-rank lifting \cite{Sherstov11Pattern},
critical block-sensitivity lifting~\cite{GP18CommunicationLowerBounds,HN12VirtueSuccinctProofs}, and lifting for monotone span programs \cite{PR17StronglyExponential, PR18LiftingNS, Robere18Thesis}.

\subsection{A New Lifting Theorem}

In this paper, we generalize a lifting theorem of Pitassi and Robere \cite{PR18LiftingNS} to use any gadget that has nontrivial rank. This theorem takes a search problem associated with an unsatisfiable CNF, and lifts a lower bound on the Nullstellensatz degree of the CNF to a lower bound on a related communication problem.

More specifically,  let $\mathcal{C}$ be an unsatisfiable \kcnf formula. The search problem associated with $\mathcal{C}$, $\Search(\mathcal{C})$, takes
as input an assignment to the underlying variables, and outputs a clause that is falsified by the assignments. \cite{PR18LiftingNS} prove that for any unsatisfiable $\mathcal{C}$, and for a sufficiently rich gadget $g$, deterministic communication complexity lower bounds for
the composed search problem $\Search(\mathcal{C})  \circ g^n$ follow from Nullstellensatz degree lower bounds for $\formf$.\footnote{In fact the result is quite a bit stronger---it applies to Razborov's rank measure~\cite{Razborov90Applications}, which is a strict strengthening of deterministic communication complexity.}
We significantly improve this lifting theorem so that it holds for \emph{any} gadget of large enough rank.

\begin{theorem}\label{cor:nss-communication-eq-intro}
Let $\mathcal{C}$ be a CNF over $n$ variables, let $\mathbb{F}$ be any field,
and let $g$ be any gadget of rank at least $r$.
Then the deterministic communication complexity of
$\Search(\mathcal{C} \circ g^n)$ is at least
$\NS_{\mathbb{F}}(\mathcal{C})$, the Nullstellensatz degree of
$\mathcal{C}$, as long as
$r \geq cn/\NS_{\mathbb{F}}(\mathcal{C})$ for some large enough constant $c$.
\end{theorem}

An important special case of our generalized theorem is when the gadget $g$ is the equality function. In this work, we apply our theorem to resolve two open problems in proof complexity and circuit complexity. Both solutions depend crucially on the ability to use the equality gadget.

We note that lifting with the equality gadget has recently been the focus of another paper. Loff and Mukhopadhyay \cite{LM19Lifting} observed that
a lifting theorem for \emph{total} \emph{functions} with the
equality gadget can be proven using a rank argument. Surprisingly, they also observed that it is \emph{not} possible to lift query complexity to communication complexity for arbitrary relations! Concretely, \cite{LM19Lifting} give an example of a relation with linear query complexity but whose composition with equality has only polylogarithmic communication complexity.
Nonetheless, they are able to prove a lifting theorem for general relations using the equality gadget by replacing standard query complexity with a stronger complexity measure (namely, the $0$-query complexity of the relation).

Unfortunately, we cannot use either of the lifting theorems of \cite{LM19Lifting} for our applications. Specifically, in our applications we lift a search problem (and therefore cannot use their result for total functions), and this search problem has small $0$-query complexity (and therefore we cannot use their lifting theorem for general relations). Indeed, this shows that our lifting theorem is incomparable to the results of \cite{LM19Lifting}, even when specialized to the equality gadget. We note that our theorem, too, bypasses the impossibility result of \cite{LM19Lifting} by using a stronger complexity measure, which in our case is the Nullstellensatz degree.

\subsection{A Separation in Proof Complexity}
The main application of our lifting theorem is the first separation in proof complexity between cutting planes proofs with high-weight versus low-weight coefficients. The cutting planes proof-system is a proof system that can be used to refute an unsatisfiable CNF by translating it into a system of integer inequalities and showing that this system has no integer solution. The latter is achieved by a sequence of steps that derive new integer inequalities from old ones, until we derive the inequality $0 \ge 1$ (which clearly has no solution). The efficiency of such a refutation is measured by its length (i.e., the number of steps) and its space (i.e., the maximal number of inequalities that have to be stored simultaneously during the derivation).

The standard variant of the cutting planes proof system, commonly denoted by CP, allows the inequalities to use coefficients of arbitrary size. However, it is also interesting to consider the variant in which the coefficients are polynomially bounded, which is commonly denoted by \cpstar. This gives rise to the natural question of the relative power of CP vs. \cpstar: are they polynomially equivalent or is there a super-polynomial length separation? This question appeared in~\cite{BC96CuttingPlanes} and remains stubbornly open to date.
In this work we finally make progress by exhibiting a setting in which unbounded coefficients afford an exponential increase in proof power.

\begin{theorem}\label{th:tradeoff-intro}
There is a family of CNF formulas of size $N$
that have cutting planes refutations of length $\bigohtilde{N^2}$ and space $O(1)$, but
for which any refutation in length $L$ and space $s$ with polynomially bounded
coefficients must satisfy $s \log L = \bigomegatilde{N}$.
\end{theorem}

Our result is the first result in proof complexity demonstrating
any situation where high-weight coefficients are more powerful than
low-weight coefficients.
In comparison, for computing Boolean functions, the relative power of high-weight and low-weight linear threshold functions has been understood for a long time.
The greater-than function can be computed by
high-weight threshold functions, but not by low-weight threshold functions, and
weights of length polynomial in $n$ suffice \cite{Muroga-book} for Boolean functions.
For higher depth threshold formulas, it is known that depth-$d$ threshold
formulas of high-weight can efficiently be computed by
depth-$(d+1)$ threshold formulas of low-weight \cite{Goldmann-threshold}.

In contrast to our near-complete knowledge of high versus low weights for functions,
almost nothing is known about the relative power of
high versus low weights in the context of proof complexity.
Buss and Clote \cite{BC96CuttingPlanes}, building on work by Cook, Coullard, and Turán \cite{CCT87ComplexityCP}, proved an analog of Muroga's result for cutting planes, showing that weights
of length polynomial in the length of the proof suffice. Quite remarkably, this result is not known to hold for
other linear threshold proof systems:
there is \emph{no} nontrivial upper bound
on the weights for more general linear threshold propositional
proof systems (such as \emph{stabbing planes}~\cite{BFIKPPR18Stabbing}, and Krajíček's threshold logic proof system~\cite{Krajicek95Frege} where one can additionally branch on linear threshold formulas).
Prior to our result,
there was no separation between high and low weights, for any
linear threshold proof system.

\subsection{A Separation in Circuit Complexity}

A second application of our lifting theorem relates to monotone real circuits,
which were introduced by Pudl\'{a}k~\cite{Pudlak97LowerBounds}.
A monotone real circuit is a generalization of monotone Boolean circuits
where each gate is allowed to compute any non-decreasing real function
of its inputs, but the inputs and output of the circuit are Boolean.
A formula is a tree-like circuit, that is, every gate has fan-out one.
The first (exponential) lower bound for monotone real circuits
was proven already in~\cite{Pudlak97LowerBounds} by
extending the lower bound for computing the clique-colouring function with monotone Boolean
circuits~\cite{Razborov85LowerBounds,AB87Monotone}.
This lower bound, together with a generalization of the interpolation
technique~\cite{Krajicek97Interpolation} which applied only to
\cpstar, was used by Pudl\'{a}k
to obtain the first exponential lower bounds for CP.

Shortly after monotone real circuits were introduced,
there was an interest in understanding the
power of monotone real computation
in comparison to monotone Boolean computation.
By extending techniques in~\cite{RM99Separation},
Bonet \etal prove that
there are functions with polynomial size monotone Boolean
circuits that require monotone real formulas of exponential size~\cite{Johannsen98Lower,BEGJ00RelativeComplexity}.
This illustrates the power of DAG-like computations in comparison
to tree-like.
In the other direction, we would like to know whether
monotone real circuits are exponentially stronger
than monotone Boolean circuits.
Rosenbloom~\cite{Rosenbloom97Monotone} presented an elegant,
simple proof that monotone real formulas are exponentially
stronger than (even non-monotone) Boolean circuits,
since slice functions can be computed by
linear-size monotone real formulas,
whereas by a counting argument we know that most slice functions
require exponential size Boolean circuits.

The question of finding explicit functions that demonstrate that
monotone real circuits are stronger than general Boolean circuits
is much more challenging since it involves proving explicit lower bounds
for Boolean circuits---a task that seems currently
completely out of reach.
A more tractable problem is that of finding explicit functions
showing that
monotone real circuits or formulas are stronger than
\emph{monotone} Boolean circuits or formulas,
but prior to this work, no such separation was known either.
We provide an explicit separation for \emph{monotone formulas}, that is, we
 provide a family of explicit functions that can be computed with monotone real
 formulas of near-linear size but require exponential monotone Boolean formulas.
This is the first explicit example that illustrates
the strength of monotone real computation.
\begin{theorem}\label{th:monotone-intro}
  There is an explicit family of functions $f_n$ over
  $\bigoh{n\polylog n}$ variables that can be computed by
  monotone real formulas of size $\bigoh{n\polylog n}$ but for which every monotone
  Boolean formula requires size $2^{\bigomega{n/\log n}}$.
\end{theorem}

Another motivation for studying lifting theorems with simple gadgets,
and in particular the equality gadget, are connections with proving \emph{non-monotone} formula size lower bounds.
As noted earlier, lifting theorems have been extremely successful in proving
monotone circuit lower bounds, and it has also been shown to be useful in some computational settings that are only ``partially'' monotone; notably monotone span programs \cite{RPRC16Exponential, PR17StronglyExponential,PR18LiftingNS} and extended formulations \cite{GJW18ExtensionComplexity, KMR17ApproximatingRectangles}.

This raises the question of to what extent lifting techniques can help prove \emph{non-monotone} lower bounds.
The beautiful work by Karchmer, Raz and Wigderson \cite{KRW95Superlogarithmic} initiated such an approach for separating $\Pclass$ from $\NC^1$---this opened up a line of research popularly known as the \emph{KRW conjecture}.
Intriguingly, steps towards resolving the KRW conjecture are
closely connected to proving lifting theorems for the equality gadget.
The first major progress was made in \cite{EIRS}
where lower bounds for the universal
relation game are proven, which is an important special case of the KRW conjecture.
Their result was recently improved in several papers \cite{GMWW17TowardBetter,HW90Composition,KM18ImprovedComposition}, and
Dinur and Meir \cite{Dinur-Meir} gave
a new top-down proof of the
state-of-the-art $\Omega(n^3)$ formula-size lower bounds via
the KRW approach.

The connection to lifting using the equality gadget is obtained by
observing that the KRW conjecture involves communication problems
in which Alice and Bob are looking for a bit on which they
differ---this is exactly an \emph{equality} problem.
Close examination of the results in \cite{EIRS,HW90Composition}
show that they are equivalent to proving lower bounds for the search problem associated with
the pebbling formula when lifted with a $1$-bit equality gadget on a \emph{particular} graph \cite{Pitassi16note}.
Our proof of \refthm{th:tradeoff} actually establishes near-optimal lower bounds on the communication complexity
of the pebbling formula lifted with equality for \emph{any} graph,
but where the size of the equality is not 1.
Thus if our main theorem could be improved
with one-bit equality gadgets this would imply the results of \cite{EIRS,HW90Composition} as a direct corollary and with significantly better parameters.

\subsection{Overview of Techniques}

We conclude this section by giving a brief overview of our techniques,
also trying to
convey some of the simplicity of the proofs which we believe is
an extra virtue of these results.

\paragraph{Lifting theorem}

In order to prove their lifting theorem, Pitassi and Robere \cite{PR18LiftingNS} defined a notion of a ``good'' gadget. They then showed that if we compose a polynomial $p$ with a good gadget $g$, the rank of the resulting matrix $p \circ g^n$ is determined \emph{exactly} by the non-zero coefficients of $p$ and the rank of $g$. Their lifting theorem follows by using this correspondence to obtain bounds on the ranks of certain matrices, which in turn yield the required communication complexity lower bound.

In this work, we observe that every gadget $g$ can be turned into a good gadget using a simple transformation. This observation allows us to get an approximate bound on the rank of $p \circ g^n$ for any $g$ with nontrivial rank. While the correspondence we get in this way is only an approximation and not an exact correspondence as in \cite{PR18LiftingNS}, it turns out that this approximation is sufficient to prove the required lower bounds. We thus get a lifting theorem that works for every gadget $g$ with sufficiently large rank.

\paragraph{Cutting planes separation}
The crux of our separation between CP and \cpstar is the following observation: CP can encode a conjunction of linear equalities with a single equality, by using exponentially large coefficients. This allows CP refutations to obtain a significant saving in space when working with \emph{linear equalities}. This saving is not available to \cpstar, and this difference between the proof systems allows the separation.

In order to exploit this observation, one of our main innovations is to concoct the separating formula.
 To do this, we must come up with a candidate formula
that can only be refuted by reasoning about a large conjunction of
linear equalities, to show that
cutting planes (CP) can efficiently refute it, and to show that low-weight cutting planes (\cpstar) cannot.

To find such a candidate formula family we resort to \emph{pebbling formulas}
which have played a major role in many proof complexity trade-off results. Interestingly,
pebbling formulas have short resolution proofs that reason in terms of large conjunctions
of literals. When we lift such formulas with the equality gadget this proof can be simulated in
cutting planes by using the large coefficients to encode many equalities with a single equality.
This yields cutting planes refutation of any pebbling formula in quadratic length and constant space.

On the other hand we prove our time-space lower bound
showing that any \cpstar refutation requires large length or large space for the same formulas.
To prove this lower bound, the first step is to instantiate the
connection in~\cite{HN12VirtueSuccinctProofs} linking time/space bounds
for many proof systems to communication complexity lower
bounds for lifted search problems.
This connection means that we can obtain the desired \cpstar-lower bounds for our
formulas $\pebblingformula \circ \eq^n$ by proving communication complexity lower bounds for the corresponding
lifted search problem $\Search(\pebblingformula) \circ \eq^n$.

In order to prove the latter communication lower bounds, we prove lower bounds on the Nullstellensatz degree of $\Search(\pebblingformula)$, and then invoke our new lifting theorem to translate them into communication lower bounds for $\Search(\pebblingformula) \circ \eq^n$.
To show the Nullstellensatz lower bounds, we prove the following lemma, which establishes an equivalence between Nullstellsatz degree and the \emph{reversible} pebbling price, and may be interesting in its own right. (We remark that connections between Nullstellensatz degree and pebbling were previously shown in \cite{BCIP02Homogenization}; however their result was not tight.)

\begin{lemma}\label{lem:pebbling-nss-intro}
For any field $\mathbb{F}$ and any directed acyclic graph $G$ the Nullstellensatz degree
of $\pebblingformula$ is equal to the reversible pebbling price of $G$.
\end{lemma}

We remark that due to known lower and upper bounds in query and proof complexity, this lemma immediately implies that Nullstellensatz degree coincides for (deterministic) decision tree and parity decision tree complexity. We record this here as a corollary, as it may be of independent interest, and provide its proof in Appendix~\ref{sec:query-complexity}.
\begin{corollary}
\label{cor:query-complexity}
For any field $\mathbb{F}$ and any directed acyclic graph $G$, the Nullstellensatz degree over $\mathbb{F}$ of $\pebblingformula$, the decision tree depth of $\Search(\pebblingformula)$, and the parity decision tree depth of $\Search(\pebblingformula)$ coincide and are equal to the reversible pebbling price of $G$.
\end{corollary}

Using the above equivalence, we obtain near-linear Nullstellensatz degree refutations for a family of graphs with maximal pebbling price,
which completes our time/space lower bound for \cpstar.
However, in order to separate CP and \cpstar we require a very specific gadget and lifting theorem.
Specifically, the gadget should
be strong enough, so that lifting holds for deterministic
communication complexity (which can efficiently simulate
small time/space \cpstar proofs), but on the other hand also weak enough,
so that lifting does not hold
for stronger communication models (randomized, real) that can efficiently
compute high-weight inequalities.
The reason that we are focusing on the equality gadget is that it
hits this sweet spot---it requires large
deterministic communication complexity, yet has short randomized
protocols, and furthermore equalities can be represented with a
single pair of inequalities.

\sfdrcomment{An extra virtue of our result is that our proof is elegant and, at least in hindsight, simple.}

\paragraph{Separation for monotone formulas}
As was the case for the separation between CP and \cpstar,
to obtain a separation between monotone Boolean formulas
and monotone real formulas we must find a function
that has just the right level of hardness.

To obtain a size lower bound for monotone Boolean
formulas we invoke the characterization of formula depth by
communication complexity of the Karchmer--Wigderson game~\cite{KW90Monotone}.
By choosing a function that has the same
Karchmer--Wigderson game as the search problem of a lifted pebbling
formula, we get a depth lower bound for monotone Boolean formulas from the
communication lower bound of the search problem.
Note that since monotone Boolean formulas can be balanced,
a depth lower bound implies a size lower bound.

In the other direction, we would like to show that these
functions are easy for real computation.
Analogously to the Karchmer--Wigderson relation,
it was shown in~\cite{HP18Note} that there is a correspondence
between
real DAG-like communication protocols (as defined in~\cite{Krajicek98Interpolation})
and monotone real circuits.
Using this relation,
a small monotone real \emph{circuit} can be extracted from
a short CP proof of the lifted pebbling formula.
However, we would like to establish a monotone real \emph{formula}
upper bound. One way to achieve this is by finding small tree-like CP
refutations of lifted pebbling formulas.
The problem is that for many gadgets lifted pebbling formulas
require exponentially long tree-like proofs. Nevertheless, for pebbling formulas lifted with the equality gadget
we are able to exhibit a short \emph{semantic} tree-like CP
refutation, which via real communication yields a small
monotone real formula.

\subsection{Organization of This Paper}
Section~\ref{sec:preliminaries} contains formal definitions of concepts discussed above and some useful facts.
Our main lifting theorem is proven in Section~\ref{sec:lifting}. Section~\ref{sec:cp} is devoted to proving our separation
between high-weight and low-weight cutting planes.
In \refsec{sec:monotone} we prove the separation between monotone real and Boolean formulas.
We conclude in Section \ref{sec:conclusion} with some open problems.

\section{Preliminaries}
\label{sec:preliminaries}

In this section we review some background material from communication
complexity and proof complexity.

\subsection{Communication Complexity and Lifted Search Problems}

Given a function $g: \mathcal{X} \times \mathcal{Y} \rightarrow \mathcal{I}$, we denote by $g^n : \mathcal{X}^n \times \mathcal{Y}^n \rightarrow \mathcal{I}^n$ the function that takes as input $n$ independent instances of $g$ and applies $g$ to each of them separately.
A \emph{total search problem} is a relation $\mathcal{S} \subseteq \mathcal{I} \times \mathcal{O}$ such that for all $z \in \mathcal{I}$ there is an
$o \in \mathcal{O}$ such that $(z, o) \in \mathcal{S}$. Intuitively, $S$ represents the computational task in which we are given an input $z \in \mathcal{I}$ and would like to find an output $o \in \mathcal{O}$ that satisfies $(z, o) \in \mathcal{S}$.

An important example of a search problem, which has proved to be very useful for proof complexity results, comes from unsatisfiable $k$-CNF formulas.
Given a $k$-CNF formula $\mathcal{C}$ over variables $z_1,\ldots,z_n$, the search problem $\Search(\mathcal{C}) \subseteq \set{0,1}^n \times \mathcal{C}$ takes as input an assignment $z \in \{0,1\}^n$ and outputs a clause $C \in \mathcal{C}$ that is falsified by $z$.

Given a search problem $\mathcal{S} \subseteq \mathcal{I}^n \times \mathcal{O}$ with a product input domain and a function $g: \mathcal{X} \times \mathcal{Y} \rightarrow \mathcal{I}$, we define the \emph{composition} $\mathcal{S} \circ g^n \subseteq \mathcal{X}^n \times \mathcal{Y}^n \times \mathcal{O}$ in the natural way: $(x, y, o) \in \mathcal{S} \circ g$ if and only if $(g^n(x, y), o) \in \mathcal{S}$.
We remark that this composition notation extends naturally to functions: for instance, if $f : \mathcal{I}^n \rightarrow \mathbb{F}$ is a function taking values in some field $\mathbb{F}$, for example, then the composition $f \circ g^n$ is a $\mathcal{X}^n \times \mathcal{Y}^n$ matrix over $\mathbb{F}$.
Second, we remark that we will sometimes write $S \circ g$ instead of $S \circ g^n$ if $n$ is clear from context.

A \emph{communication search problem} is a search problem with a bipartite input domain $\mathcal{I} = \mathcal{A} \times \mathcal{B}$.
A communication protocol for a search problem $\mathcal{S} \subseteq \mathcal{A} \times \mathcal{B} \times \mathcal{O}$ is a strategy for a collaborative game where two players Alice and Bob hold $x \in \mathcal{A}, y \in \mathcal{B}$, respectively, and wish to output an $o \in \mathcal{O}$ such that $\left( (x,y), o \right) \in \mathcal{S}$ while communicating as few bits as possible.
Messages are sent sequentially until one player announces the answer and only depend on the input of one
player and past messages. The cost of a protocol is the maximum number
of bits sent over all inputs, and the communication complexity of a
search problem, which we denote by $\Pcc(\mathcal{S})$, is the minimum cost over all protocols that solve $\mathcal{S}$.
For more details on communication complexity, see, \eg,~\cite{KN97CommunicationCplx}.

Given a CNF formula $\formf$ on $n$ variables $z_1, z_2, \ldots, z_n$ and a Boolean function
$\gadgetname\colon \set{0,1}^q \times \set{0,1}^q \to \set{0,1}$, we
define a \emph{lifted formula} $\liftedformula{\formf}{\gadgetname^n}$ as follows.
For each variable $z_i$ of $\formf$, we
have $2 q$ new variables $x_{i,1}, \ldots , x_{i,q}, y_{i,1}, \ldots, y_{i,q}$.
For each clause $\clc \in \formf$ we replace each literal $z_i$ or $\neg z_i$ in $\clc$ by a CNF
encoding of either $\gadgetname(x_{i,1}, \ldots , x_{i,q}, y_{i,1}, \ldots, y_{i,q})$ or $\olnot{\gadgetname}(x_{i,1}, \ldots , x_{i,q}, y_{i,1}, \ldots, y_{i,q})$ according
to the sign of the literal. We then expand the resulting expression
into a CNF, which we denote by $\liftedformula{\clc}{\gadgetname}$, using de~Morgan's
rules. The substituted formula is
$\liftedformula{\formf}{\gadgetname} = \Union_{\clc \in
  \formf}\liftedformula{\clc}{\gadgetname}$.

For the sake of an example, consider the clause $u \vee \overline v$, and we will substitute with the equality gadget on two bits. 
Formally, we replace $u$ with $x_{u,1}x_{u,2} = y_{u,1}y_{u,2}$ and $v$ with $x_{v,1}x_{v,2} = y_{v,1}y_{v,2}$. 
We can encode a two-bit equality as the CNF formula \[ (x_1x_2 = y_1y_2) \equiv (\overline x_1 \vee y_1) \wedge (x_1 \vee \overline y_1) \wedge (\overline x_2 \vee y_2) \wedge (x_2 \vee \overline y_2), \]
and a two-bit disequality as the CNF formula \[ (x_1x_2 \neq y_1y_2) \equiv (\overline x_1 \vee \overline x_2 \vee \overline y_1 \vee \overline y_2) \wedge (\overline x_1 \vee x_2 \vee \overline y_1 \vee y_2) \wedge (x_1 \vee \overline x_2 \vee y_1 \vee \overline y_2) \wedge (x_1 \vee x_2 \vee y_1 \vee y_2). \]
So, in the clause $u \vee \overline v$, we would substitute $u$ for the CNF encoding of $x_{u,1}x_{u,2} = y_{u,1}y_{u,2}$ and $\overline v$ with the CNF encoding of $x_{v,1}x_{v,2} \neq y_{v,1}y_{v,2}$; finally, we would convert the new formula to a CNF by distributing the top $\vee$ over the $\wedge$s from the new CNF encodings.

While $\Search(\mathcal{C}) \circ \gadgetname^n$ is not the same problem
as $\Search(\formf \circ \gadgetname^n)$, we can reduce the former to the latter.
Specifically, suppose we are given a protocol $\Pi$ for $\Search(\formf \circ \gadgetname^n)$.
Consider the following protocol $\Pi'$ for $\Search(\mathcal{C}) \circ \gadgetname^n$:
Given an input $(x,y)$, the protocol $\Pi'$ interprets $(x,y)$ as an input to $\Pi$.
Now, assume that $\Pi'$ outputs on $(x,y)$ a clause $\cld$ of  $\liftedformula{\formf}{\gadgetname^n}$,
which was obtained from a clause $\clc$  of $\formf$. Then, the clause $\clc$ is a valid $\Search(\mathcal{C})$
on $(x,y)$, so $\Pi'$ outputs it.
Let us record this observation.

\begin{observation}
  \label{obs:search-lift-lift-search}
  $\Pcc(\Search(\formf \circ \gadgetname)) \geq \Pcc(\Search(\mathcal{C}) \circ \gadgetname)$ for any unsatisfiable CNF $\mathcal{C}$ and any Boolean gadget $g$.
\end{observation}

\subsection{Nullstellensatz}

As a proof system, Nullstellensatz allows verifying that a set of
polynomials does not have a common root, and it can also be used to
refute CNF formulas by converting them into polynomials. It plays an
important role in our lower bounds.

  Let $\mathbb{F}$ be a field, and let $\mathcal{P} = \set{p_1 = 0, p_2 = 0, \ldots, p_m = 0}$ be an unsatisfiable system of polynomial equations in $\mathbb{F}[z_1, z_2, \ldots, z_n]$.
  A \emph{Nullstellensatz refutation} of $\mathcal{P}$ is a sequence of polynomials $q_1, q_2, \ldots, q_m \in \mathbb{F}[z_1, z_2, \ldots, z_n]$ such that $ \sum_{i=1}^m p_iq_i = 1$ where the equality is syntactic.
  The \emph{degree} of the refutation is $\max_{i} \deg(p_i q_i)$; the \emph{Nullstellensatz degree} of $\mathcal{P}$, denoted $\NS_\mathbb{F}(\mathcal{P})$, is the minimum degree of any Nullstellensatz refutation of $\mathcal{P}$.

Let $\mathcal{C} = C_1 \wedge C_2 \wedge \cdots \wedge C_m$ be an unsatisfiable CNF formula over Boolean variables $z_1, z_2, \ldots, z_n$. 
We introduce a standard encoding of each clause $C_i$ as a polynomial equation.
If $C$ is a clause then let $C^+$ denote the set of variables occurring positively in $C$ and $C^-$ denote the set of variables occurring negatively in $C$; with this notation we can write $ C = \bigvee_{z \in C^+} z \vee \bigvee_{z \in C^-}  \overline{z}. $
From $C$ define the polynomial \[\mathcal{E}(C) \equiv \prod_{z \in C^+}(1- z) \prod_{z \in C^-}z, \]  over formal variables $z_1, z_2, \ldots, z_n$.
Observe that $\mathcal{E}(C) = 0$ is satisfied (over $0/1$ assignments to $z_i$) if and only if the corresponding assignment satisfies $C$.
We abuse notation and let $\mathcal{E}(\mathcal{C}) = \set{\mathcal{E}(C) : C \in \mathcal{C}} \cup \set{z_i^2 - z_i}_{i \in [m]},$ and note that the second set of polynomial equations restricts the $z_i$ inputs to $\set{0,1}$ values.
The \emph{$\mathbb{F}$-Nullstellensatz degree} of $\mathcal{C}$, denoted $\NS_{\mathbb{F}}(\mathcal{C})$, is the Nullstellensatz degree of refuting $\mathcal{E}(\mathcal{C})$.

How do we know that a Nullstellensatz refutation always exists?
One can deduce this from Hilbert's Nullstellensatz, but for our
purposes it is enough to use a simpler version proved by Buss et
al. (Theorem 5.2 in~\cite{BIKPRS97ProofCplx}): if $\mathcal{P}$ is a
system of polynomial equations over $\mathbb{F}[z_1, \ldots, z_n]$
with no $\set{0,1}$ solutions, then there exists a Nullstellensatz
refutation of $\mathcal{P} \cup \set{z_i^2 - z_i = 0}_{i \in [n]}$.

\subsection{Cutting Planes}

The \introduceterm{Cutting planes (CP)} proof system
was introduced in~\cite{CCT87ComplexityCP}
as a formalization of the integer linear programming algorithm in
\cite{Gomory63AlgorithmIntegerSolutions,Chvatal73EdmondPolytopes}.
Cutting planes proofs give a formal method to deduce new linear inequalities from old that are \emph{sound over integer solutions}---that is, if some integral vector $x^*$ satisfies a set of linear inequalities $\mathcal{I}$, then $x^*$ will also satisfy any inequality $ax \geq b$ deduced from $\mathcal{I}$ by a sequence of cutting planes deductions.
The allowed deductions in a cutting planes proof are the following:
\[     \text{Linear combination}
\ \
\AxiomC{$\sum_i a_i x_i \geq A $}
\AxiomC{$ \sum_i b_i x_i \geq B$}
\BinaryInfC{$\sum_i \, (c a_i +  d b_i) x_i \geq cA + dB$}
\DisplayProof 
\quad \quad
\text{Division}
\AxiomC{$\sum_i c a_i x_i \geq A$}
\UnaryInfC{$\sum_i a_i x_i \geq \ceiling{A/c}$}
\DisplayProof
\]
where
$a_i$, $b_i$, $c$, $d$, $A$, and~$B$ are all integers
and $c, d \geq 0$.

In order to use cutting planes to refute unsatisfiable CNF formulas, we need to translate clauses to inequalities.
It is easy to see how to do this by example: we translate the clause
$x \lor y \lor \olnot{z}$
to the inequality
$x + y + (1-z) \geq 1$,
or, equivalently, 
$x + y -z \geq 0$
if we collect all constant terms on the right-hand side.
For refuting CNF formulas we equip cutting planes proofs with the following additional rules ensuring all variables take $\set{0,1}$ values:

\[
\text{Variable axioms}
\ \ 
\AxiomC{\rule{0pt}{8pt}}
\UnaryInfC{$\ x \geq 0$}
\DisplayProof
\ \ \
\AxiomC{\rule{0pt}{8pt}}
\UnaryInfC{$\ -x \geq -1$}
\DisplayProof
\]

The goal, then, is to prove unsatisfiability by deriving the inequality
$0 \geq 1$. 
This is possible if and only if there is no
$\set{0,1}$-assignment  satifying all constraints.

As discussed in the introduction, we are interested in several natural parameters of cutting planes proof---length, space, and the sizes of the coefficients.
So, we define a cutting planes refutation as a sequence of \emph{configurations} (this is also known as the \emph{blackboard model}).
A configuration is a set of linear inequalities with
integer coefficients, and a sequence of configurations
$\clsc_0,\ldots,\clsc_\lengthstd$ is a cutting planes refutation of a
formula $\formf$ if $\clsc_0=\emptyset$, $\clsc_\lengthstd$ contains
the contradiction $0 \geq 1$, and each configuration $\clsc_{t+1}$
follows from $\clsc_{t}$ either by adding an inequality in $\formf$, by
adding the result of one of the above inference rules where all the premises
are in $\clsc_{t}$, or by removing an inequality present in $\clsc_{t}$.
The length of a refutation is then defined to be the number of configurations $L$; the space\footnote{Formally, this is known as the \emph{line space}.} is $\max_{t\in[L]}\setsize{\clsc_t}$, the maximum number of inequalities in a configuration; and
the coefficient bit size is the maximum size in bits of a coefficient
that appears in the refutation.

For any proof system, it is natural to ask what is the minimal amount of space needed to prove tautologies.
Indeed, there has been much work in the literature studying this, and for proof systems such as resolution (e.g.~\cite{ET01SpaceBounds,ABRW02SpaceComplexity,BG03SpaceComplexity,BN08ShortProofs}) and polynomial calculus (e.g.~\cite{ABRW02SpaceComplexity,FLNRT15SpaceCplx,BG15Framework,BBGHMW17SpaceProofCplx}) it is known that there are unsatisfiable CNF formulas which unconditionally require large space to refute.
In contrast (and quite surprisingly!) it was shown
in~\cite{GPT15SpaceComplexityCP} that for cutting planes proofs, constant line space is always enough.  The proof presented in
\cite{GPT15SpaceComplexityCP} does use coefficients of exponential
magnitude, but the authors are not able to show that this is
necessary---only that coefficients of at most constant magnitude are
not sufficient.

Similarly, one can ask 
whether cutting planes refutations require large
coefficients to realize the full power of the proof system.
Towards this, define \cpstar to be cutting planes proofs with
\emph{polynomially-bounded coefficients} or, in other words, a cutting planes refutation $\Pi$ of a formula $\mathcal{C}$ with $n$ variables is a \cpstar refutation if the largest coefficient in $\Pi$ has magnitude $\mathsf{poly}(n, L)$.

The question of how \cpstar relates to unrestricted cutting planes
has been raised in several papers, e.g., \cite{BPR97LowerBoundsCP,BEGJ00RelativeComplexity}.
This question was studied already in~\cite{BC96CuttingPlanes},
where it was proven that any cutting planes refutation in length~$L$
can be transformed into a refutation with
$L^{\bigoh{1}}$ lines having coefficents of
magnitude
$\exp(\bigoh{L})$ (here the asymptotic notation hides a mild
dependence on the size of the coefficients in the input).
The authors write, however, that their original goal had been to show
that coefficients of only polynomial magnitude would be enough,
\ie that \cpstar would be as powerful as cutting planes except possibly for
a polynomial loss, but that they had to leave this as an open problem.
To the best of our knowledge, there has not been
a single example of
\emph{any unsatisfiable formula} where \cpstar could potentially
perform much worse than general (high-weight) cutting planes.

Finally, as observed in \cite{BPS07LS,HN12VirtueSuccinctProofs}, we can use an
efficient cutting planes refutation of a formula $\formf$ to solve
$\Search(\formf)$ by an efficient communication protocol. Since the first configuration $\clsc_0$ is always
true and the last configuration $\clsc_\lengthstd$ is always false,
the players can simulate a binary search by evaluating the truth value
of a configuration according to their joint assignment and find a true
configuration followed by a false configuration. It is not hard to see
that the inequality being added corresponds to a clause in $\formf$
and it is a valid answer to $\Search(\formf)$.

\begin{lemma}[\cite{HN12VirtueSuccinctProofs}]
  \label{lem:cp-tradeoff-reduction}
  If there is a cutting planes refutation of $\formf$ in length
  $\lengthstd$, line space $\spacestd$, and coefficient bit size
  $\coefficientbits$, then there is a deterministic communication
  protocol for $\Search(\formf)$ of cost
  $\bigoh{\spacestd(\coefficientbits+\log\numvariables)\log\lengthstd}$.
  \end{lemma}

\section{Rank Lifting from Any Gadget}
\label{sec:lifting}

In this section we discuss our new lifting theorem, restated next.\footnote{In fact, we prove a somewhat more general theorem (see Theorem~\ref{th:ns-lifting} in Appendix \ref{sec:nullstellensatz-lifting} for details). We also remark that this theorem in fact holds for a stronger communication measure (Razborov's \emph{rank measure} \cite{Razborov90Applications}), and so implies lower bounds for other models---see Appendix \ref{sec:nullstellensatz-lifting} for details.}

\begin{theorem}
  \label{cor:nss-communication-eq}
  Let $\mathcal{C}$ be any unsatisfiable $k$-CNF on $n$ variables and let $\mathbb{F}$ be any field.
  For any Boolean valued gadget $g$ with $\rank(g) \geq 12enk/\NS_{\mathbb{F}}(\mathcal{C})$ we have \[ \Pcc(\Search(\mathcal{C}) \circ g) \geq \NS_{\mathbb{F}}(\mathcal{C}).\]
\end{theorem}

This generalizes a recent lifting theorem from \cite{PR18LiftingNS}, which only allowed certain ``good'' gadgets. The main technical step of that proof showed that ``good'' gadgets can be used to lift the degree of multilinear polynomials  to the rank of matrices. In this section, we improve this, showing that \emph{any} gadget with non-trivial rank can be used to lift polynomial degree to rank. Given this result, Theorem \ref{cor:nss-communication-eq} is proved by reproducing the proof of \cite{PR18LiftingNS} with a tighter analysis.
With this in mind, in this section we will prove our new lifting argument for degree to rank, and then relegate the rest of the proof of Theorem \ref{cor:nss-communication-eq} to Appendix \ref{sec:nullstellensatz-lifting}.

Let us now make these arguments formal.
We start by recalling the definition of a ``good" gadget of \cite{PR18LiftingNS}.

\begin{definition}[Definition 3.1 in \cite{PR18LiftingNS}]\label{def:good-gadget}
  Let $\mathbb{F}$ be a field.
  A gadget $g: \mathcal{X} \times \mathcal{Y} \rightarrow \mathbb{F}$ is \emph{good} if for any matrices $A, B$ of the same size we have \[ \rank(\mathds{1}_{\mathcal{X}, \mathcal{Y}} \otimes A + g \otimes B) = \rank(A) + \rank(g)\rank(B)\] where $\mathds{1}_{\mathcal{X}, \mathcal{Y}}$ denotes the $\mathcal{X} \times \mathcal{Y}$ all-$1$s matrix.
\end{definition}

In \cite{PR18LiftingNS} it is shown that good gadgets are useful because they lift \emph{degree} to \emph{rank} when composed with multilinear polynomials.

\begin{theorem}[Theorem 1.2 in \cite{PR18LiftingNS}]
  Let $\mathbb{F}$ be any field, and let $p \in \mathbb{F}[z_1, z_2, \ldots, z_n]$ be a multilinear polynomial over $\mathbb{F}$.
  For any good gadget $g: \mathcal{X} \times \mathcal{Y} \rightarrow \mathbb{F}$ we have \[ \rank(p \circ g^n) = \sum_{S: \hat p(S) \neq 0} \rank(g)^{|S|}.\]
\end{theorem}

In the present work, we show that a gadget being good is \emph{not} strictly necessary to obtain the above lifting from degree to rank.
In fact, composing with \emph{any} gadget lifts degree to rank!

\begin{theorem}
  \label{th:rank-lifting}
  Let $p \in \mathbb{F}[z_1, z_2, \ldots, z_n]$ be any multilinear polynomial and let $g : \mathcal{X} \times \mathcal{Y} \rightarrow \mathbb{F}$ be any non-zero gadget with $\rank(g) \geq 3$.
  Then \[ \sum_{S: \hat p(S) \neq 0} (\rank(g) - 3)^{|S|} \leq \rank(p \circ g^n) \leq \sum_{S: \hat p(S) \neq 0} \rank(g)^{|S|}.\]
 \end{theorem}

We remark that the lower bound in the theorem can be sharpened to $\rank(g) - 2$ if the gadget $g$ is not full rank.
While the previous theorem does not require the gadget $g$ to be good, the notion of a good gadget will still play a key role in the proof. The general idea is that every gadget with non-trivial rank can be transformed into a good gadget with a slight modification.
With this in mind, en-route to proving Theorem \ref{th:rank-lifting} we give the following characterization of good gadgets which may be of independent interest.

\begin{lemma}
  \label{lem:alternative-good}
  A gadget $g$ is good if and only if the all-$1$s vector is not in the row or column space of $g$.
\end{lemma}

In the remainder of the section we prove Theorem \ref{th:rank-lifting} and Lemma \ref{lem:alternative-good}.

\subsection{Proof of Lemma \ref{lem:alternative-good}}
\label{sec:proof-alternative-good}

We begin by proving Lemma \ref{lem:alternative-good}, which is by a simple linear-algebraic argument.
Given a matrix $M$ over a field, let $row(M)$ denote the row-space of $M$ and let $col(M)$ denote the column-space of $M$.
The following characterization of when rank is additive will be crucial.

\begin{theorem}[\cite{MS72Rank}]\label{th:additive-rank-criterion}
  For any matrices $A, B$ of the same size over any field, $\rank(A + B) = \rank(A) + \rank(B)$ if and only if $row(A) \cap row(B) = col(A) \cap col(B) = \set{\mathbf{0}}$.
\end{theorem}

The previous theorem formalizes the intuition that rank should be additive if and only if the corresponding linear operators act on disjoint parts of the vector space.
Using the previous theorem we deduce the following general statement, from which Lemma \ref{lem:alternative-good} immediately follows.

\begin{lemma}\label{lem:general-criterion}
  Let $f, g$ be matrices over any fixed field $\mathbb{F}$ of the same size.
  The following are equivalent:
  \begin{enumerate}
    \item For all matrices $A, B$ of the same size, $\rank(f \otimes A + g \otimes B) = \rank(f)\rank(A) + \rank(g)\rank(B).$
    \item $\rank(f + g) = \rank(f) + \rank(g).$
  \end{enumerate}

\end{lemma}
\begin{proof}
  By choosing $A = B = (1)$ we instantly deduce (2) from (1). To prove the converse, we use Theorem \ref{th:additive-rank-criterion}.
  Let $A, B$ be matrices such that $\rank(f \otimes A + g \otimes B) \neq \rank(f)\rank(A) + \rank(g)\rank(B)$.
  Then by Theorem \ref{th:additive-rank-criterion} it follows that there is a non-zero vector in the intersection of either the row- or column-spaces of $f \otimes A$ and $g \otimes B$.
  Suppose that there is a non-zero vector $u \in col(f \otimes A) \cap col(g \otimes B)$, and we prove that there is a non-zero vector in $col(f) \cap col(g)$ implying $\rank(f + g) \neq \rank(f) + \rank(g)$.
  (A symmetric argument will apply to the row spaces.)

  Assume that $f$ and $g$ are $a \times b$ dimensional matrices, and that $A$ and $B$ are $m \times n$ dimensional matrices. Let $u$ be the length $am$ non-zero vector in the column spaces of both $f \otimes A$ and $g \otimes B$, and suppose without loss of generality that $u_1 \neq 0$. It follows that there are length $bn$ vectors $x, y$ such that $(f \otimes A)x = u = (g \otimes B)y$.
  Write
  \begin{align*}
    x & = (x^1, x^2, \ldots, x^b), \\
    \quad y & = (y^1, y^2, \ldots, y^b)
  \end{align*}
  where $x^i, y^i$ are vectors of length $n$ for each $i$.

  Let $A_1$ denote the first row of $A$ and $B_1$ denote the first row of $B$; note they are both vectors of length $n$.
  Define the length-$b$ vectors
  \begin{align*}
    x' & = (A_1 x^1, A_1 x^{2}, \ldots, A_1 x^{b}), \\
    y' & = (B_1 y^1, B_1 y^2, \ldots, B_1 y^b).
  \end{align*}
  Then, by definition, for each $i = 1, 2, \ldots, a$ we have $(fx')_i = u_{(i-1)m + 1} = (gy')_i$, and the vector is non-zero since $u_1 \neq 0$ by assumption.
  Thus $fx' = gy'$ and the column spaces of $f$ and $g$ intersect at a non-zero vector.
\end{proof}

From Lemma \ref{lem:general-criterion} we can deduce Lemma \ref{lem:alternative-good} immediately.

\begin{proof}[Proof of Lemma \ref{lem:alternative-good}]
  By the previous lemma, $g$ is good if and only if $\rank(\mathds{1} + g) = \rank(\mathds{1}) + \rank(g)$.
  By Theorem \ref{th:additive-rank-criterion} this is true iff the all-$1$s vector is not in the row- or column-space of $g$.
\end{proof}

\subsection{Proof of Theorem \ref{th:rank-lifting}}
\label{sec:proof-rank-lifting}

In this section we prove Theorem \ref{th:rank-lifting} using Lemma \ref{lem:alternative-good}.
The theorem follows by induction using the following lemma, and the proof mimics the proof from \cite{PR18LiftingNS, Robere18Thesis}.

\begin{lemma}
  Let $\mathbb{F}$ be any field, and let $g: \mathcal{X} \times \mathcal{Y} \rightarrow \mathbb{F}$ be any gadget with $\rank(g) \geq 3$.
  For any matrices $A, B$ of the same size we have \[ \rank(\mathds{1}_{\mathcal{X}, \mathcal{Y}} \otimes A + g \otimes B) \geq \rank(A) + (\rank(g) - 3)\rank(B)\] where $\mathds{1}_{\mathcal{X}, \mathcal{Y}}$ is the $\mathcal{X} \times \mathcal{Y}$ all-$1$s matrix.
\end{lemma}

\begin{proof}
  Assume without loss of generality that $|\mathcal{X}| \geq |\mathcal{Y}|$ and let $\mathds{1} = \mathds{1}_{\mathcal{X}, \mathcal{Y}}$.
  Thinking of $g$ as a matrix, let $u$ be any column vector of $g$.
  If we zero the entries of $u$ in $g$, then the remaining matrix cannot have full rank, implying that some row-vector $v$ of the remaining matrix will become linearly dependent.
  Let $g_1$ be the $\mathcal{X} \times \mathcal{Y}$ matrix consisting of the $u$ column and $v$ row of $g$, and let $g_2$ be the $\mathcal{X} \times \mathcal{Y}$ matrix obtained by zeroing out $u$ and $v$ in $g$.
  Observe $g = g_1 + g_2$, and also since $g_2$ contains an all-$0$ row and an all-$0$ column it is good by Lemma \ref{lem:alternative-good} (as any linear combination of rows/columns of $g$ must contain a zero coordinate).

  Now, observe that
  \begin{align*}
    \rank(\mathds{1} \otimes A + g \otimes B) & = \rank(\mathds{1} \otimes A + g_1 \otimes B + g_2 \otimes B) \\
                                              & \geq \rank(\mathds{1} \otimes A + g_2 \otimes B) - \rank(g_1 \otimes B) \\
                                              & = \rank(\mathds{1} \otimes A + g_2 \otimes B) - \rank(g_1)\rank(B)
  \end{align*}
  where the inequality follows since adding a rank-$R$ matrix can decrease the rank by at most $R$.
  Since $g_1$ consists of a single non-zero row and column we have $\rank(g_1) \leq 2$; by the construction of $g_2$ we have $\rank(g_2) = \rank(g) - 1$.
  Using these facts and the fact that $g_2$ is good, we have
  \begin{align*}
    \rank(\mathds{1} \otimes A + g_2 \otimes B) - \rank(g_1)\rank(B) & \geq \rank(A) + \rank(g_2)\rank(B) - 2\,\rank(B) \\
                                                                     & = \rank(A) + (\rank(g) - 3)\rank(B). \qedhere
  \end{align*}
\end{proof}

With the lemma in hand we can prove Theorem \ref{th:rank-lifting}.

\begin{proof}[Proof of Theorem \ref{th:rank-lifting}]
  We prove \[ \rank(p \circ g^n) \geq \sum_{S: \hat p(S) \neq 0} (\rank(g) - 3)^{|S|}\] by induction on $n$, the number of variables.

  Observe that the inequality is trivially true if $n = 0$.
  Assume $n > 0$, and let $\mathds{1} = \mathds{1}_{\mathcal{X}, \mathcal{Y}}$.
  Write $p = q + z_1r$ for multilinear polynomials $q, r \in \mathbb{F}[z_2, z_3, \ldots, z_n]$. Note that it clearly holds that $p \circ g^n = \mathds{1} \otimes (q \circ g^{n-1}) + g \otimes (r \circ g^{n-1})$.
  From the claim we have by induction that
  \begin{align*}
    \rank(p \circ g^n) & = \rank(\mathds{1} \otimes (q \circ g^{n-1}) + g \otimes (r \circ g^{n-1})) \\
                       & \geq \rank(q \circ g^{n-1}) + (\rank(g)-3)\rank(r \circ g^{n-1}) \\
                       & = \sum_{S: \hat q(S) \neq 0} (\rank(g)-3)^{|S|} + (\rank(g)-3)\sum_{T: \hat r(T) \neq 0} (\rank(g)-3)^{|T|} \\
                       & = \sum_{S: \hat p(S) \neq 0 \atop z_1 \not \in S} (\rank(g) - 3)^{|S|} + (\rank(g)-3)\sum_{T: \hat p(T) \neq 0 \atop z_1 \in T} (\rank(g)-3)^{|T| - 1} \\
                       & = \sum_{S: \hat p(S) \neq 0} (\rank(g) - 3)^{|S|}.
  \end{align*}
  For the upper bound, by subadditivity of rank we have
 \begin{align*}
 \rank(\mathds{1} \otimes A + g \otimes B) &\leq \rank(\mathds{1} \otimes A) + \rank(g \otimes B) \\
						& =\rank(\mathds{1}) \rank(A) + \rank(g)\rank(B) \\
						& = \rank(A) + \rank(g)\rank(B).
\end{align*}
  Apply the above induction argument using this inequality \emph{mutatis mutandis}.
\end{proof}

\section{Application: Separating Cutting Planes Systems}
\label{sec:cp}

In this section we prove a new separation between high-weight and low-weight cutting planes proofs in the bounded-space regime.
\begin{theorem}
  \label{th:tradeoff}
  There is a family of $\bigoh{\log\log n}$-CNF formulas over
  $\bigoh{n\log\log n}$ variables and $\bigohtilde{n}$ clauses
  that have CP refutations in length $\bigohtilde{n^2}$ and line
  space $\bigoh{1}$, but for which any \cpstar refutation in length
  $\lengthstd$ and line space $\spacestd$ must satisfy
  $\spacestd\log\lengthstd = \bigomega{n/\log^2 n}$.
\end{theorem}

By the results of \cite{GPT15SpaceComplexityCP}, \emph{any} unsatisfiable CNF formula has a cutting planes refutation in constant line space, albeit with exponential length and exponentially large coefficients. In Theorem \ref{th:tradeoff} we show that the length of such a refutation can be reduced to polynomial for certain formulas, described next.

At a high level, we prove \refth{th:tradeoff} using the reversible pebble game. Given any DAG $G$ with a unique sink node $t$, the \introduceterm{reversible pebble game}~\cite{Bennett89TimeSpaceReversible} is a single-player game that
is played with a set of pebbles on $G$. Initially
the graph is empty, and at each step the player can either place or
remove a pebble on a vertex whose predecessors already have pebbles
(in particular the player can always place or remove a pebble on a source).
The goal of the game is to place a pebble on the sink while using as few pebbles as possible.
The reversible pebbling price of a graph, denoted $\rpeb(G)$, is the minimum number of pebbles required to place a pebble on the sink.

The family of formulas witnessing \refth{th:tradeoff} are \emph{pebbling
formulas} composed with the equality gadget. Intuitively, the pebbling formula~\cite{BW01ShortProofs} $\pebblingformula$ associated with $G$ is a formula that claims that it is impossible to place a pebble on the sink (using any number of pebbles). Since it is always possible to place a pebble by using some amount of pebbles, this formula is clearly a contradiction.

Formally, the pebbling formula $\pebblingformula$ is the following CNF formula.
For each vertex $u \in V$ there is a variable $z_u$ (intuitively, $z_u$ should take the value ``true" if and only if it is possible to place a pebble on $u$ using any number of pebbles). The variables are constrained by the following clauses.
\begin{itemize}
\item a clause $z_s$ for each source vertex $s$ (i.e., we can always place a pebble on any source),
\item a clause $\Lor_{u \in \prednode{v}} \olnot z_u \lor z_v$ for each non-source vertex $v$ with predecessors $\prednode{v}$ (i.e., if we can place a pebble on the predecessors of $v$, then we can place a pebble on $v$), and
\item a clause $\olnot{z_t}$ for the sink $t$ (i.e., it is impossible to place a pebble on $t$).
\end{itemize}

Proving Theorem \ref{th:tradeoff} factors into two tasks: a lower
bound and an upper bound. By applying our lifting theorem from the
previous section, the lower bound will follow immediately from a good
lower bound on the Nullstellensatz degree of pebbling formulas. In
order to prove lower bounds on the Nullstellensatz degree, we show in
Section \ref{sec:nsatz} that over \emph{every} field, the
Nullstellensatz degree required to refute $\pebblingformula$ is
\emph{exactly} the reversible pebbling price of $G$. We then use it
together with our lifting theorem to prove the time-space tradeoff for
bounded-coefficient cutting planes refutations of
$\pebblingformula \circ g$ in Section \ref{sec:lowerbounds} for
\emph{any} high-rank gadget $g$. Finally, in Section
\ref{sec:upper-bounds} we prove the upper bound by presenting a short
and constant-space refutation of $\pebblingformula \circ \equalgadget$
in cutting planes with unbounded coefficients.

\subsection{Nullstellensatz Degree of Pebbling Formulas}
\label{sec:nsatz}

In this section we prove that the Nullstellensatz degree of the pebbling formula of a graph $G$ equals the reversible pebbling price of $G$.
\begin{lemma}
  \label{lem:pebbling-nss}
  For any field $\mathbb{F}$ and any graph $G$, $\NS_{\mathbb{F}}(\pebblingformula) = \rpeb(G)$.
\end{lemma}

We crucially use the following dual characterization of Nullstellensatz degree by designs 
\cite{Buss98LowerBoundsNS}.

\begin{definition}\label{def:ns-designs}
  Let $\mathbb{F}$ be a field, let $d$ be a positive integer, and let $\mathcal{P}$ be an unsatisfiable system of polynomial equations over $\mathbb{F}[z_1, z_2, \ldots, z_n]$.
  A \emph{$d$-design} for $\mathcal{P}$ is a linear functional $D$ on the space of polynomials satisfying the following axioms:
  \begin{enumerate}
  \item $D(1) = 1$.
  \item For all $p \in \mathcal{P}$ and all polynomials $q$ such that $\deg(pq) \leq d$, we have $D(pq) = 0$.
  \end{enumerate}
\end{definition}

Clearly, if we have a candidate degree-$d$ Nullstellensatz refutation $1 = \sum p_iq_i$, then applying a $d$-design to both sides of the refutation yields $1 = 0$, a contradiction.
Thus, if a $d$-design exists for a system of polynomials then there cannot be a Nullstellensatz refutation of degree $d$.
Remarkably, a converse holds for systems of polynomials over $\set{0,1}^n$.

\begin{theorem}[Theorems 3, 4 in \cite{Buss98LowerBoundsNS}]\label{thm:ns-design}
  Let $\mathbb{F}$ be a field and let $\mathcal{P}$ be a system of polynomial equations over $\mathbb{F}[z_1, z_2, \ldots, z_n]$ containing the Boolean equations $z_i^2 - z_i = 0$ for all $i \in [n]$.
  Then $\mathcal{P}$ does not have a degree-$d$ Nullstellensatz refutation if and only if it has a $d$-design.
\end{theorem}

With this characterization in hand we prove \reflem{lem:pebbling-nss}. 

\begin{proof}[Proof of \reflem{lem:pebbling-nss}]
  Let $G$ be a DAG, and consider the pebbling formula $G$.
  Following the standard translation of CNF formulas into unsatisfiable systems of polynomial equations, we express $\pebblingformula$ with the following equations:
  \begin{description}
  \item[Source Equations.] The equation $(1 - z_s) = 0$ for each source vertex $s$.
  \item[Sink Equations.] The equation $z_t = 0$ for the sink vertex $t$.
  \item[Neighbour Equations.] The equation $(1 - z_v)\prod_{u \in \mathrm{pred}(v)} z_u = 0$ for each internal vertex $v$.
  \item[Boolean Equations.] The equation $z_v^2 - z_v = 0$ for each vertex $v$.
  \end{description}
We prove that a $d$-design for the above system exists if and only if $d < \rpeb(G)$, and this implies the lemma. Let $D$ be a $d$-design for the system.
First, note that since the Boolean axioms are satisfied and since $D$ is linear, it follows that $D$ is completely specified by its value on multilinear monomials $z_T := \prod_{i \in T}z_i$ (with this notation note that $z_\emptyset := 1$). Moreover, $D$ must satisfy the following properties:
\begin{description}
  \item[Empty Set Axiom.] $D(z_{\emptyset}) = 1$.
  \item[Source Axioms.] $D(z_T) = D(z_Tz_s)$ for every source $s$ and every $T \subseteq [n]$ with $|T \cup \set{s}| \leq d$.
  \item[Neighbour Axioms.] $D(z_Tz_{\mathrm{pred}(v)}) = D(z_Tz_{\mathrm{pred}(v)}z_v)$ for every non-source vertex $v$ and every $T \subseteq [n]$ with $|T \cup \mathrm{pred}(v)\cup \set{v}| \leq d$. 
  \item[Sink Axiom.] $D(z_T z_t) = 0$ for the sink $t$ and every $T \subseteq [n]$ with $|T \cup \set{t}| \leq d$.
  \end{description}
  We may assume without loss of generality that $D(z_T) = 0$ for any set $T$ with $|T| > d$.

  Given a set $S$ of vertices of $G$, we think of $S$ as the reversible pebbling configuration in which there are pebbles on the vertices in $S$ and there are no pebbles on any other vertex. We say that a configuration $T$ is \introduceterm{reachable} from a configuration $S$ if there is a sequence of legal reversible pebbling moves that changes $S$ to $T$ while using at most $d$ pebbles at any given point.

 Now, we claim that the only way to satisfy the first three axioms is to set $D(x_T) = 1$ for every configuration $T$ that is reachable from $\emptyset$. To see why, observe that those axioms are satisfiable if and only if the empty configuration is assigned the value $1$, any configuration containing the sink is labelled $0$, and $D(z_S) = D(z_T)$ for any two configurations $S, T$ with at most $d$ pebbles that are mutually reachable via a single reversible pebbling move. Hence, this setting of $D$ is the only one we need to consider.

  Finally, observe that this specification of a design $D$ satisfies the sink axiom if and only if $d < \rpeb(G)$, since the sink is reachable from $\emptyset$ using $\rpeb(G)$ pebbles but not with less (by the definition of $\rpeb(G)$). Therefore, a $d$-design for $\pebblingformula$ exists if and only if $d < \rpeb(G)$, as required. 
 \end{proof}

\subsection{Time-Space Lower Bounds for Low-Weight Refutations}
\label{sec:lowerbounds}

In this section we prove the lower bound part of the time-space
trade-off for \cpstar.
\begin{lemma}
  \label{lem:cpstar-is-hard}
  There is a family of graphs $\set{G_n}$ with $n$ vertices and constant degree, such that every \cpstar refutation of
  $\pebblingformula[G_n] \circ \eq$ in length $\lengthstd$ and line space $\spacestd$
  must have $\spacestd\log\lengthstd = \bigomega{n/\log^2 n}$.
\end{lemma}

Our plan is to lift a pebbling formula that is hard with respect to
Nullstellensatz degree, and as we just proved it is enough to find a
family of graphs whose reversible pebbling price is large.
Paul \etal~\cite{PTC76SpaceBounds} provide such a family (and in fact 
prove their hardness in the stronger standard pebbling model).
\begin{theorem}\label{th:hard-graphs}
  There is a family of graphs $\set{G_n}$ with $n$ vertices, constant degree, and for which $\rpeb(G_n) = \bigomega{n/\log n}$.
\end{theorem}

We combine these graphs with our lifting theorem as follows.

\begin{lemma}
  \label{lem:search-is-hard}
  There is a family of graphs $\set{G_n}$ with $n$ vertices and
  constant degree, such that
  $\Pcc(\Search(\pebblingformula\circ\eq))= \bigomega{n/\log n}$.
\end{lemma}

\begin{proof}
  Let $\pebblingformula$ be the pebbling formula of a graph $G=G_n$
  from the family given by \refth{th:hard-graphs}.
  By \reflem{lem:pebbling-nss} the Nullstellensatz degree of the
  formula is
\begin{equation}
  \NS_{\mathbb{F}}(\pebblingformula) = \rpeb(G) = \bigomega{n/\log
    n} \eqperiod
\end{equation}

  This allows us to use our lifting theorem,
  \refth{cor:nss-communication-eq}, with an equality gadget of arity
  $\gadgetarity=\bigoh{\log (n/\NS_{\mathbb{F}}(\pebblingformula)}=\bigoh{\log\log n}$, and
  obtain that the lifted search problem
  $\Search(\pebblingformula)\circ\eq$ requires deterministic
  communication
\begin{equation}
  \Pcc(\Search(\pebblingformula)\circ\eq) \geq
  NS_{\mathbb{F}}(\pebblingformula) = \bigomega{n/\log n} \eqperiod
\end{equation}

  As we noted in \refobs{obs:search-lift-lift-search}, this implies
  that the search problem of the lifted formula also requires
  deterministic communication
\begin{equation}
  \Pcc(\Search(\pebblingformula\circ\eq)) \geq
  \Pcc(\Search(\pebblingformula)\circ\eq) = \bigomega{n/\log n} \eqperiod
\end{equation}
\end{proof}

Since we collected our last ingredient, let us finish the proof.

\begin{proof}[Proof of \reflem{lem:cpstar-is-hard}]
  Let $S=\Search(\pebblingformula \circ \eq)$ be the search problem
  given by \reflem{lem:search-is-hard}.
Using \reflem{lem:cp-tradeoff-reduction} we have that every
  cutting planes refutation of the lifted formula in length
  $\lengthstd$, line space $\spacestd$, and coefficient length
  $\coefficientbits$ must satisfy
  \begin{equation}
    \label{eq:almost-there}
  \spacestd(\coefficientbits+\log\numvariables)\log\lengthstd =
  \bigomega{\Pcc(S)} =
  \bigomega{n/\log n} \eqperiod
\end{equation}

Since the size of the lifted formula
$\pebblingformula\circ\eq$ is $\bigohtilde{n}$, the
coefficients of a \cpstar refutation are bounded by a polynomial of
$n$ in magnitude, and hence by $\bigoh{\log n}$ in
length. Substituting the value of $\coefficientbits=\bigoh{\log n}$ in
\eqref{eq:almost-there} we obtain that
\begin{equation}
  \spacestd\log\lengthstd = \bigomega{n/\log^2 n}
\end{equation}
  as we wanted to show.
\end{proof}

\subsection{Time-Space Upper Bounds for High Weight Refutations}
\label{sec:upper-bounds}
\label{sec:small-space}

We now prove \refth{th:small-space}, showing that cutting planes proofs with large coefficients can efficiently refute pebbling formulas composed with equality gadgets in constant line space.
Let $\equalgadget_\gadgetarity$ denote the equality gadget on $q$ bits.

\begin{theorem}
  \label{th:small-space}
  Let $\pebblingformula$ be any constant-width pebbling formula.
  There is a cutting planes refutation of $\liftedformula{\pebblingformula}{\eqloglog}$ in length $\tilde{O}(\numvariables^2)$ and space $\bigoh{1}$.
\end{theorem}

We also use the following lemma, which is a ``derivational'' analogue of the recent result of~\cite{GPT15SpaceComplexityCP} showing that any set of unsatisfiable integer linear inequalities
has a cutting planes refutation in constant space.
As the techniques are essentially the same we leave the proof to Appendix \ref{sec:space-lemma}.

\begin{lemma}[Space Lemma]
  \label{lem:clause-in-space-five}
  Let $\formf$ be a set of integer linear inequalities over $\numvariables$
  variables that implies a clause $\clc$. Then there is a cutting
  planes derivation of $\clc$ from $\formf$ in length
  $\bigoh{\numvariables^2 2^\numvariables}$ and space $\bigoh{1}$.
\end{lemma}

Let us begin by outlining the high level proof idea.
We would like to refute the lifted formula $\liftedformula{\pebblingformula}{\equalgadget_q}$ using constant space. Consider first the unlifted formula $\pebblingformula$.
The natural way to refute it is the following: Let $v_1, \ldots, v_n$ be a topological ordering of the vertices of $G$. The refutation will go over the vertices in this order, each time deriving the equation that says that the variable $z_{v_i}$ must take the value ``true" by using the equations that were derived earlier for the predecessors of $v_i$. Eventually, the refutation will derive the equation that says that the sink must take the value ``true", which contradicts the axiom that says that the sink must be false.

Going back to the lifted formula $\liftedformula{\pebblingformula}{\equalgadget_q}$, we construct a refutation using the same strategy, except that now the equation of $z_{v_i}$ is replaced with the equations
\[ x_{v_i, 1} = y_{v_i, 1}, \ldots x_{v_i, q} = y_{v_i, q}. \]

The main obstacle is that if we implement this refutation in the naive way, we will have to store all the equations simultaneously, yielding a refutation of space $\bigoh{q\cdot n}$. The key idea of our proof is that CP can encode the conjunction of many equations using a \emph{single} equation. We can therefore use this encoding in our refutation to store at any given point all the equations that were derived so far in a single equation. The implementation yields a refutation of constant space, as required.

To see how we can encode multiple equations using a single equation, consider the following example. Suppose we wish to encode the equations
\[ x_1 = y_1, x_2 = y_2, x_3 = y_3, \]
where all the variables take values in $\set{0,1}$. Then, it is easy to see that those equations are equivalent to the equation
\[ 4 \cdot x_1 + 2 \cdot x_2 + x_3 = 4 \cdot y_1 + 2 \cdot y_2 + y_3. \]
This idea generalizes in a straightforward way to deal with more equations, as well as with arbitrary linear gadgets, to be discussed below.

  The rest of this section is devoted to the proof of Theorem~\ref{th:small-space}.
  The following notion is central to the proof.
  Say a gadget $g(x, y) : \set{0,1}^q \times \set{0, 1}^q \rightarrow \set{0,1}$ is \emph{linear} if there exists a linear expression with integer coefficients \[L(x, y) = c + \sum_{i=1}^q a_ix_i + b_iy_i \] such that $g(x, y) = 1$ if and only if $L(x, y) = 0$.
  Note that the equality gadget is linear, as it corresponds to the linear expression $\sum_{i=1}^q 2^{i-1}(x_i - y_i)$.

  Let $g$ be any linear gadget with corresponding linear expression $L$.
  Let $K = 1 + \max_{x, y} |L(x, y)|$, and let $G$ be the underlying DAG of the composed pebbling formula $\pebblingformula \circ g^n$.
  Note that for each vertex $u$ of $G$ the composed formula has corresponding variables $x_u, y_u \in \set{0,1}^q$.
  Once and for all, fix an ordering of the vertices of $G$ and assume that all subsets are ordered accordingly.
  For a subset of vertices $U \subseteq V$ define \[ L(U) := \sum_{u_i \in U} K^i L(x_{u_i}, y_{u_i}).\]
  The following claim shows that $L(U)$ encodes the truth of the conjunction $\bigwedge_{u_i \in U} g(x_{u_i}, y_{u_i})$.

\begin{claim}
For a set of vertices $U$ and any $x, y \in \set{0,1}^{\gadgetarity\numvertices}$, $L(U) = 0$ if and only if $\bigwedge_{u_i \in U} g(x_{u_i}, y_{u_i}) = 1$.
\end{claim}
\begin{proof}
    Since $g$ is linear, if $g(x_{u_i}, y_{u_i}) = 1$ for all $u_i \in U$ then it follows that $L(x_{u_i}, y_{u_i}) = 0$ for all $u_i \in U$, which in turn implies $L(U) = 0$.
    Conversely, suppose $g(x_{u_i}, y_{u_i}) = 0$ for some vertex $u_i$, and let $i$ be the \emph{largest} such index.
    It follows that $\lingadget{u_i} \neq 0$, and clearly
    \begin{equation}
      \ABS{\sum_{j < i} \safetyconstant^{j}\lingadget{x_{u_j}, y_{u_j}}} \leq
      \sum_{j<i} \safetyconstant^{j}\abs{\lingadget{x_{u_j}, y_{u_j}}} \leq
      (\safetyconstant-1)\sum_{j<i} \safetyconstant^{j}
      < \safetyconstant^i.
    \end{equation}
    This implies $\lingadget{U}\neq 0$, since
	\begin{align*}
		\ABS{L(U)} &= \ABS{K^i \cdot L(u_i) + \sum_{j < i} \safetyconstant^{j}\lingadget{x_{u_j}, y_{u_j}}} \\
			&\ge \ABS{K^i \cdot L(u_i)} - \ABS{\sum_{j < i} \safetyconstant^{j}\lingadget{x_{u_j}, y_{u_j}}} \\
			&\ge K^i - \ABS{\sum_{j < i} \safetyconstant^{j}\lingadget{x_{u_j}, y_{u_j}}} > 0 \qedhere
	\end{align*}
  \end{proof}

  From here on in the proof, we consider $L(U) = 0$, or $L(U)$ for short, as being syntactically represented in cutting planes as the pair of inequalities $L(U) \geq 0$, $-L(U) \geq 0$.
  The bulk of the proof lies in the following lemma, which shows how to ``encode'' and ``decode'' unit literals in the expressions $L(U)$.

  \begin{lemma}[Coding Lemma]\label{lem:coding-lemma}
    Let $U$ be any set of vertices.
    Then
    \begin{enumerate}
    \item For any $u \in U$ there is a cutting planes derivation of $L(u)$ from $L(U)$ in length $O(q\setsize{U})$ and space $O(1)$. \label{item:coding-extract}
    \item Let $C = \neg z_{u_1} \vee \neg z_{u_2} \vee \cdots \vee \neg z_{u_{k-1}} \vee z_{u_k}$ be an axiom of $\pebblingformula$ with $u_1, u_2, \ldots, u_{k-1} \in U$.
    Let $\ell_g$ and $s_g$ be such that there exists a derivation of $L(u)$ from a CNF encoding of $g(u)$ in length $\ell_g$ and space $s_g$.
      From $L(u_1), L(u_2), \ldots, L(u_{k-1})$ and $C \circ g^n$ there is a cutting planes derivation of $L(u_k)$ in length $O(2^{kq} \ell_g)$ and space $O(s_g)$. \label{item:coding-propagate}
      \item For any $u \not \in U$ there is a cutting planes derivation of $L(U \cup \set{u})$ from $L(U)$ and $L(u)$ in length $O(1)$ and space $O(1)$. \label{item:coding-append}
    \end{enumerate}
  \end{lemma}

Let us first use the Coding Lemma to complete the proof. We show a
more general statement from which \refth{th:small-space} follows
immediately by setting $k=3$ and $g=\equalgadget_q$, with
$\gadgetarity=\bigoh{\log \log n}$, and bounding
$\ell_{\equalgadget} = \bigoh{q}$ and $s_{\equalgadget}= \bigoh{1}$.

\begin{lemma}
  \label{lem:small-space}
  If $\pebblingformula$ is a width-$k$ pebbling formula on $n$ variables and $g$ is a linear gadget of arity $\gadgetarity$ then there is a
  cutting planes refutation of $\liftedformula{\pebblingformula}{g^n}$ in
  length $O(n(kqn + 2^{kq}\ell_g))$ and space $\bigoh{\clwidth+s_g}$.
\end{lemma}
\begin{proof}
  We begin with $L(\emptyset)$, which is represented as the pair of inequalities $0 \geq 0, 0 \geq 0$.
  By combining Parts (2) and (3) of the Coding Lemma we can derive $L(S)$, where $S$ is the set of sources of $G$.
  We then follow a unit-propagation proof of $\pebblingformula$, deriving $L(u)$ for each vertex of $G$ in topological order.
  Suppose at some point during the derivation we have derived $L(U)$ for some subset $U$ of vertices.
  For any axiom $C$ of $\pebblingformula$ of the form $C = \neg z_{u_1} \vee \neg z_{u_2} \vee \cdots \vee \neg z_{u_{k-1}} \vee z_{u_k}$ with $u_1, u_2, \ldots, u_{k-1} \in U$ we do the following: first apply Part (1) of the Coding Lemma to obtain $L(u_i)$ for each $i \in [k-1]$.
  Apply Part (2) to derive $L(u_k)$, forget $L(u_i)$ for each $i \in [k-1]$, and then apply Part (3) to $L(U)$ and $L(u_k)$ to derive $L(U \cup \set{u_k})$.
  Continue in this way until we derive $L(z)$ where $z$ is the sink vertex of $G$.
  Since $\set{L(z), \olnot{z} \circ g}$ is an unsatisfiable set of linear inequalities, it follows by the Space Lemma (Lemma \ref{lem:clause-in-space-five}) that we can deduce a contradiction in length $O(q^22^q)$ and space $O(1)$.

  In the above proof we need to derive $L(u)$ for each of the $n$ vertices of the graph.
  Deriving $L(u)$ requires at most $O(k)$ applications of Part (1), one application of Part (2), and one application of Part (3).
  Thus, in total, we require length $O(n(kqn + 2^{kq}\ell_g))$ and space $O(k + s_g)$.
\end{proof}

It remains to prove the Coding Lemma (\reflem{lem:coding-lemma}).
  \begin{proof}[Proof of Coding Lemma]
    Let $U = \set{u_1, u_2, \ldots, u_t}$ be an arbitrary subset of vertices of size $t$.
    Recall the definition $L(U) = \sum_{i=1}^t K^{i-1}L(u_i)$.
    For any $u_i \in U$ a \emph{term} of $L(U)$ will be one of the terms $K^{i-1}L(u_i)$, which is a sum of $2q$ variables itself.
    We begin by defining two auxiliary operations that allow us to trim both the least and the most significant terms from $L(U)$.

To
    trim the $\trimctant$ least significant terms of an inequality we
    essentially divide by $\safetyconstant^\trimctant$. More formally, for
    every variable $v$ with a positive coefficient
    $\varcoefficient$ less than $\safetyconstant^\trimctant$ we add the inequality
    $-\varcoefficient \vargeneric \geq -\varcoefficient$, and
    for every variable with a negative coefficient greater than
    $-\safetyconstant^\trimctant$ we add the inequality $\varcoefficient\vargeneric \geq 0$.
    This takes length $\bigoh{\gadgetarity\trimctant}$, since each
    term contributes $2q$ coefficients, and space $\bigoh{1}$.

    At this point all the remaining coefficients on the LHS are
    divisible by $\safetyconstant^\trimctant$, so we can apply the
    division rule.  By construction the RHS is greater than
    $-\safetyconstant^\trimctant - \sum_{j
      \geq\trimctant}\freecoefficient K^{j}$, therefore when we divide
    by $\safetyconstant^\trimctant$ the coefficient on the RHS becomes
    $-\sum_{j\geq\trimctant}\freecoefficient
    \safetyconstant^{j-\trimctant}$.

    Finally, to restore the values of the coefficients to the values
    they had before dividing, we multiply by
    $\safetyconstant ^\trimctant$ at the end to restore them.

To trim
    the $\ineqsize-\trimctant$ most significant terms of an inequality
    with $\ineqsize$ terms we need to use the opposite inequality, since
    the remaining part only has a semantic meaning when the most
    significant part vanishes. Hence we first trim the $\trimctant$ least
    significant terms of the opposite inequality, keeping exactly the
    negation of the terms that we want to discard. Then we add both
    inequalities so that only the $\trimctant$ least significant terms
    remain. This takes length $\bigoh{\gadgetarity\trimctant}$ and space $\bigoh{1}$.

    Using the trimming operations we can prove items \ref{item:coding-extract}--\ref{item:coding-append} in the lemma.

\begin{enumerate}
\item We must show that for any $u \in U$ there is a cutting planes derivation of $L(u)$ from $L(U)$ in length $O(qt)$ and space $O(1)$.
    This is straightforward: begin by making copies of the pair of inequalities $L(U) \geq 0$ and $-L(U) \geq 0$ encoding $L(U)=0$. Trim the terms that are strictly more and strictly less significant than $L(u)$ from both of the inequalities, in length $\bigoh{\gadgetarity t}$ and space $\bigoh{1}$.
    
\item Recall that
    we assumed there is a derivation $\Pi$ of $L(u_k)$ from the CNF
    formula $z_{u_k} \circ g$ in length $\ell_g$ and space $s_g$, so
    our goal is to produce the set of clauses $z_{u_k} \circ g$. Any
    such clause $D$ is implied by the set of inequalities
    $\set{L(u_i)}_{i=1}^{k-1}$ together with the CNF encoding of
    $C \circ g^n$, therefore it has a derivation $\Pi_D$ in length
    $O(2^{k\gadgetarity})$ and space $\bigoh{1}$ by the Space Lemma
    (\reflem{lem:clause-in-space-five}). Replacing each usage of a clause
    $D \in z_{u_k} \circ g$ in $\Pi$ as an axiom by the corresponding
    derivation $\Pi_D$ we obtain a sound derivation $L(u_{t+1})$ in
    length $O(2^{t\gadgetarity}\encodingtime)$ and space
    $\bigoh{\encodingspace}$.
\item Simply add $K^{t+1}L(u) \geq 0$ to $L(U) \geq 0$ and $-K^{t+1}L(u) \geq 0$ to $-L(U)\geq 0$; this clearly uses bounded length and space.\qedhere
\end{enumerate}
    \end{proof}
  \noindent
  This completes the proof of \refth{th:small-space}.

Note that the largest coefficient used in the refutation is bounded by
$\safetyconstant^{\numvariables}$. Indeed, the argument can be generalized to
give a continuous trade-off between the size of the largest coefficient and
the number of inequalities, simply by adding a new pair of empty
inequalities once the coefficient required to add a vertex to an
existing pair would be too large. This means that if we allow up to
$\extrainequalities$ inequalities then we can use coefficients of size
bounded by
$\safetyconstant^{\bigoh{\numvariables/\extrainequalities}}$.

\section{Application: Separating Monotone Boolean and Real Formulas}
\label{sec:monotone}

In this section we exhibit an explicit function that exponentially
separates the size of monotone Boolean formulas and monotone real
formulas.

\begin{theorem}
  \label{th:monotone}
  There is an explicit family of functions $f_n$ over
  $\bigoh{n\polylog n}$ variables that can be computed by 
  monotone real formulas of size $\bigoh{n\polylog n}$ but for which every monotone
  Boolean formula requires size $2^{\bigomega{n/\log n}}$.
\end{theorem}

To prove the lower bound part of \refth{th:monotone} we use the
characterization of formula
depth by communication complexity~\cite{KW90Monotone}. Given a
monotone Boolean function $f$, the monotone Karchmer--Wigderson game
of $f$ is a search problem
$\mkw\colon \set{0,1}^n \times \set{0,1}^n \to [n]$ defined as
$((x,y),i)\in\mkw$ if $f(x)=1$, $f(y)=0$, $x_i=1$, and $y_i=0$. In
other words, given a $1$-input $x$ and a $0$-input $y$ for $f$, the
task is to find an index $i\in[n]$ such that $x_i=1$, and
$y_i=0$. Such an index always exists because $f$ is monotone.

If we denote by $\mdepth$ the minimum depth of a monotone Boolean
formula required to compute a Boolean function $f$, then we can write
the characterization as
\begin{lemma}[\cite{KW90Monotone}]
  \label{lem:kw}
  For every function $f$, it holds that $\mdepth=\Pcc(\mkw)$.
\end{lemma}

The analogue of this characterization for real circuits is in terms of
DAG-like real
protocols~\cite{Krajicek98Interpolation,Sokolov17Daglike,HP18Note}. Since
we are only interested in formulas rather than circuits we only
consider tree-like protocols, which we call \emph{locally real}
protocols to distinguish them from the stronger model of real
protocols, also known as real games~\cite{Krajicek98Interpolation}.

A locally real communication protocol, then, is a communication
protocol where the set of inputs compatible with a node is defined by
one half-space, as opposed to a real protocol where the set of
compatible inputs is defined by the intersection of all half-spaces in
the path leading to that node.

Formally, a locally real protocol for a search problem
$\Search \colon \mathcal{X} \times \mathcal{Y} \to \mathcal{Z}$, where
$\mathcal{X}$ and $\mathcal{Y}$ are Boolean hypercubes, is a tree
where every internal node $v$ is labelled with a half-space
$H_v = \setdescr{(x,y)\in \mathcal{X}\times\mathcal{Y}}{\langle xy,c_v \rangle
  \geq d_v}$, where
$c_v \in \mathbb{R}^{\dim \mathcal{X} + \dim \mathcal{Y}}$ and
$d_v \in \mathbb{R}$, and every leaf is additionally labelled with an
element $z\in \mathcal{Z}$. The root is labelled with the full space
$\mathcal{X} \times \mathcal{Y}$, children are consistent in the sense that if a node $w$
has children $u$ and $v$ then $H_w \subseteq H_u \union H_v$.  Given
an input $(x,y)$, the protocol produces a nondeterministic output $z$
as follows. We start at the root and at each internal node we
nondeterministically move to a child that contains $(x,y)$, which
exists by the consistency condition. The output of the protocol is the
label of the resulting leaf. A protocol is correct if for any input
$(x,y)\in \mathcal{X}\times \mathcal{Y}$ it holds that
$z\in \mathcal{Z}$.

It is not hard to turn a real formula
into a locally real protocol, and the converse also holds.

\begin{lemma}[\cite{HP18Note}]
  \label{lem:proto-to-ckt}
  Given a locally real protocol for the monotone Karchmer--Wigderson
  game of a partial function $f$, there exists a monotone real formula
  with the same underlying graph that computes $f$.
\end{lemma}

In order to obtain a function whose Karchmer--Wigderson game we can
analyse we use the fact that every search problem can be interpreted
as the Karchmer--Wigderson game of some function. To state the result
we need the notion of a nondeterministic communication protocol, which
is a collection $N$ of deterministic protocols such that
$((x,y),z) \in S$ if and only if there exists some protocol $\pi\in N$
such that $\pi(x,y)=z$. The cost of a nondeterministic protocol is
$\log \setsize{N} + \max_{\pi\in N}\depth{\pi}$.

\begin{lemma}[\cite{Razborov90Applications,Gal01Characterization}, see also \cite{Robere18Thesis}]
  \label{lem:kw-from-search}
  Let $S$ be a two-party total search problem with nondeterministic
  communication complexity $k$. There exists a partial monotone
  Boolean function $f\colon \set{0,1}^{2^k} \to \set{0,1,*}$ such that
  $S$ is exactly the monotone Karchmer--Wigderson game of $f$.
\end{lemma}

We use as a search problem the falsified clause search problem of a
hard pebbling formula composed with equality given by
\reflem{lem:search-is-hard}. To exhibit a real formula for the
function it induces, we first build a tree-like cutting planes proof
of small size of the composed pebbling formula.

\begin{theorem}
  \label{th:tree-like}
  If $\formf$ is the pebbling formula of a graph of indegree $2$, then
  there is a tree-like semantic cutting planes refutation of
  $\liftedformula{\formf}{\eq_{\log\log \numvertices}}$ in length
  $\bigoh{\numvertices \log\numvertices\log\log\numvertices}$.
\end{theorem}

It is not hard to see that we can extract an efficient locally real
protocol from a tree-like cutting planes refutation of small size, but
let us record this fact formally.

\begin{lemma}[Folklore, see \cite{Sokolov17Daglike}]
  \label{lem:proof-to-proto}
  Given a semantic cutting planes refutation of a formula $F$, there
  is a locally real protocol for $\Search(F)$ with the same underlying graph.
\end{lemma}

Before we move into the proof of \refth{th:tree-like}, let us complete
the proof of \refth{th:monotone}.

\begin{proof}[Proof of \refth{th:monotone}]
  Let $S=\Search(\pebblingformula \circ \eqloglog)$ be the search
  problem given by \reflem{lem:search-is-hard}. The nondeterministic
  communication complexity of $f$ is
  $\log(\setsize{\pebblingformula \circ \eqloglog})+2$, since given a
  certificate consisting of a clause falsified by the inputs each
  party can independently verify that their part is falsified and
  communicate so to the other party. Therefore by
  \reflem{lem:kw-from-search} there is a partial monotone function
  $f^*$ over $\bigoh{n\polylog n}$ variables whose monotone
  Karchmer--Wigderson game is equivalent to $S$. By
  \refth{th:tree-like} there is a semantic cutting planes refutation
  of the formula $\pebblingformula \circ \eqloglog$ of length
  $\bigoh{n\polylog n}$, which we convert into a locally real protocol
  for $S$ of size $\bigoh{n\polylog n}$ using
  \reflem{lem:proof-to-proto}, and then into a monotone real formula
  for $f^*$ of size $\bigoh{n\polylog n}$ using
  \reflem{lem:proto-to-ckt}. Add a threshold gate on top of the
  formula to ensure that the output is always Boolean and let $f$ be
  the total function that the formula computes. Since $f$ extends
  $f^*$, by \reflem{lem:kw} and \reflem{lem:search-is-hard} $f$
  requires monotone Boolean formulas of depth $\bigomega{n/\log n}$,
  and therefore size $2^{\bigomega{n/\log n}}$.
\end{proof}

\subsection{A Short Tree-like Refutation}
\label{sec:tree-like}

For simplicity in this section we reinterpret the pebbling formula of
a graph $G$ of indegree $2$ lifted with equality of $\gadgetarity$
bits as the pebbling formula of a graph $G'$ lifted with equality of
$1$ bit or $\xnor$, where $G'$ is the graph where we replace every
vertex in $G$ by a blob of $\gadgetarity$ vertices and we replace
every edge by a bipartite complete graph between blobs, and with the
difference that instead of having axioms asserting that all sinks are
false, the axioms assert that some sink is false.

Without further ado, let us prove \refth{th:tree-like}, which follows by setting $q=\log\log n$ in the following Lemma.
\begin{lemma}
  \label{lem:tree-like}
  If $\formf$ is the pebbling formula of a graph of indegree $2$, then
  there is a tree-like semantic cutting planes refutation of
  $\liftedformula{\formf}{\eq_{\gadgetarity}}$ in length
  $\bigoh{\numvertices \gadgetarity 2^\gadgetarity}$.
\end{lemma}
As in \refsec{sec:small-space} we fix a topological order of $G$ and
we build a refutation by keeping two inequalities
$\lingadget{\vertexseq}\geq 0$ and $-\lingadget{\vertexseq} \geq
0$. The main difference is that we cannot use the Coding Lemma to
isolate the value of a single vertex, since then we would lose the
information on the rest of vertices, therefore we have to simulate the
inference steps in place as we describe next.

Let us set up some notation. If $\vertexseq$ is a set of vertices, let
$\boolgadget{\vertexseq}=\Land_{\vargeneric \in \vertexseq}
\xnor(\vargeneric)$. We represent $\xnor(\vargeneric)$ with
$\lingadget{\vargeneric}=0$, where
$\lingadget{\vargeneric}=\alicevar{\vargeneric}{}-\bobvar{\vargeneric}{}$,
and $\boolgadget{\vertexseq}$ with $\lingadget{\vertexseq} = 0$, where
$\lingadget{\vertexseq} = \sum_{\vargeneric_j \in \vertexseq}
2^j\lingadget{\vargeneric_j}$. We begin with $\vertexseq=\emptyset$
and with the trivial inequalities $0=\lingadget{\emptyset}=0$. Let us show how to derive each vertex.

\begin{lemma}
  \label{lem:tree-like-induction}
  There is a tree-like semantic derivation of
  $\lingadget{\vertexseq \union \set{\succvertex}} \geq 0$ from
  $\lingadget{\vertexseq} \geq 0$ and the axioms in $2^\gadgetarity$
  steps.
\end{lemma}

\begin{proof}
If $\succvertex$ is a source, then the inequality
$\lingadget{\succvertex}\geq 0$ is already an axiom, hence it is enough
to multiply $\lingadget{\vertexseq}\geq 0$ by $2$ and add
$\lingadget{\succvertex}$.

The complex case is when $\succvertex$ has predecessors
$\predvertex_1,\ldots,\predvertex_\gadgetarity$. Let
$\makelit{\vargeneric}{\polarity} = \polarity + (-1)^\polarity
\vargeneric$ be the literal over variable $\vargeneric$ and polarity
$1-\polarity$. Consider the $2^\gadgetarity$ axioms
$\axiom{\binarystring} \geq 0$ indexed by
$\binarystring \in \set{0,1}^\gadgetarity$ and defined as
\begin{align}
  \axiompremise &= \sum_{j=1}^\gadgetarity \makelit{\alicevar{u}{j}}{\binarystring_j} + \makelit{\bobvar{u}{j}}{\binarystring_j}\\
  \axiom{\binarystring} &= \axiomconclusion + \axiompremise
\eqperiod
\end{align}

We start with an inequality $\lingadget{\vertexseq} \geq 0$. In order
to have enough working space for the axioms we multiply the inequality
by $2^\gadgetarity$, and using weakening axioms we add a slack term
defined as
\begin{equation}
  \slackterm = \sum_{j=1}^\gadgetarity 2^{j-1} \alicevar{u}{j}
\end{equation}
to obtain $\lingadgetpart{0} \geq 0$ with
\begin{equation}
  \lingadgetpart{0} = 2^\gadgetarity \lingadget{\vertexseq} + \slackterm \eqperiod
\end{equation}
The coefficients for $\slackterm$ are chosen so that if we evaluate
$\slackterm$ on a string $\binarystring\in\set{0,1}^\gadgetarity$, the
result is $\binarystring$ interpreted as a binary number. We use it to keep
track of which axioms we have processed so far, in a similar fashion
to how the space-efficient refutation of the complete tautology~\cite{GPT15SpaceComplexityCP} that we reproduce in \refapp{sec:space-lemma} keeps track of processed truth value
assignments.

The next step is to add each axiom to $\lingadgetpart{0}$, but for
this to work we need to represent each intermediate step with one
inequality as follows.

\begin{claim}
We can represent the Boolean expression
\begin{equation}
  \boolgadgetpart{\binarylimit}=\ib{\lingadget{\vertexseq}\geq 0} \land
  \left(\boolgadget{\vertexseq} \limpl \Land_{\binarystring \leq \binarylimit} \axiom{\binarystring}\right)
\end{equation}
with the inequality $\lingadgetpart{\binarylimit} \geq 0$ defined as
\begin{equation}
  \lingadgetpart{\binarylimit} = (\binarylimit+1)\axiomconclusion + \lingadgetpart{0} \eqperiod
\end{equation}
\end{claim}

\begin{proof}
Let us begin proving the claim by showing that
$\boolgadgetpart{\binarylimit} \Rightarrow
\lingadgetpart{\binarylimit} \geq 0$. First consider an assignment
$\alpha$ that satisfies $\lingadget{\vertexseq} \geq 0$ but not
$\boolgadget{\vertexseq}$, that is an assignment where
$\alicevar{u}{j}=1$ and $\bobvar{u}{j}=0$ for some predecessor
$\predvertex_j$ of $\succvertex$.  Then
$\restrict{\lingadgetpart{0}}{\alpha}\geq 2^\gadgetarity$, hence
$\restrict{\lingadgetpart{\binarylimit}}{\alpha} \geq
-(\binarylimit+1)+2^\gadgetarity \geq 0$.

Now consider an assignment $\alpha$ that satisfies
$\boolgadget{\vertexseq}$, hence $\lingadget{\vertexseq}=0$. If
$\alpha_{\predvertex_1,\ldots,\predvertex_\gadgetarity} =
\binarystring \leq \binarylimit$ then, since $\alpha$ falsifies
$\axiompremise\geq 1$, $\alpha$ must satisfy $\axiomconclusion\geq 0$,
so both $\lingadgetpart{0} \geq 0$ and $\axiomconclusion \geq
0$. Otherwise if
$\alpha_{\predvertex_1,\ldots,\predvertex_\gadgetarity} =
\binarystring > \binarylimit$ then
$\restrict{\slackterm}{\alpha} = \binarystring \geq \binarylimit+1$,
and we have
$\restrict{\lingadgetpart{\binarylimit}}{\alpha} \geq
-(\binarylimit+1) + \binarystring \geq 0$.

Let us finish by showing that
$\boolgadgetpart{\binarylimit} \Leftarrow
\lingadgetpart{\binarylimit} \geq 0$. First consider an
assignment $\alpha$ that falsifies $\lingadget{\vertexseq} \geq 0$. Then
$\restrict{\lingadgetpart{\binarylimit}}{\alpha}\leq (\binarylimit+1) -2^\gadgetarity < 0$.

Now consider an assignment $\alpha$ that satisfies
$\boolgadget{\vertexseq}$ but not an axiom $\axiom{\binarystring}$
with $\binarystring\leq \binarylimit$. Then in particular $\alpha$
falsifies $\axiomconclusion\geq 0$, hence
$\restrict{\lingadgetpart{\binarylimit}}{\alpha} =
(\binarylimit+1)\axiomconclusion + \restrict{\slackterm}{\alpha} =
-(\binarylimit+1) - \binarystring < 0$. This concludes the proof of
the claim.
\end{proof}

Since $\lingadgetpart{\binarylimit+1} \geq 0$ follows semantically
from $\lingadgetpart{\binarylimit} \geq 0$ and $\axiom{\binarylimit}$,
we can derive $\boolgadgetpart{2^\gadgetarity} \geq 0$ from
$\lingadgetpart{0} \geq 0$ and the set of axioms
$\axiom{\binarystring} \geq 0$ using $2^\gadgetarity$ semantic inferences of arity
$2$. Also, $\lingadgetpart{2^\gadgetarity} \geq 0$ is semantically
(but not syntactically) equivalent to
$\lingadget{\vertexseq \union \set{\succvertex}} \geq 0$, so we can be
ready for the next step with a semantic inference of arity $1$.
\end{proof}

We can derive the upper bound inequality
$-\lingadget{\vertexseq \union \set{\succvertex}}\geq 0$ similarly,
the main differences being that we start with
$-\lingadget{\vertexseq} \geq 0$ and that we use the other half of the
axioms, that is
$\axiom{\binarystring} = -\axiomconclusion + \axiompremise$.

We handle the sinks in a slightly different way. Instead of using the
pebbling axioms directly, we first use the pebbling axioms of all the
sinks together with the axioms enforcing that some sink is false in
order to derive a set of inequalities similar to pebbling axioms but
with $-1$ in place of $\axiomconclusion$. We then use the same
derivation as in \reflem{lem:tree-like-induction} using these
inequalities in place of the axioms and we obtain
$\lingadget{\vertexseq} -1 \geq 0$ and analogously
$-\lingadget{\vertexseq} -1 \geq 0$. Adding both inequalities leads to
the contradiction $-2 \geq 0$.

To conclude the proof it is enough to observe that we do
$\bigoh{2^\gadgetarity}$ inference steps for each vertex in $G'$,
which has order $\numvertices\gadgetarity$, hence the total length of
the refutation is $\bigoh{\numvertices\gadgetarity2^\gadgetarity}$.

\section{Concluding Remarks}
\label{sec:conclusion}

In this paper, we show that the cutting planes proof system (CP) is stronger than
its variant with polynomially bounded coefficients (\cpstar) with respect to simultaneous length
and space. This is the first result in proof complexity demonstrating any 
situation where high-weight coefficients are more powerful than 
low-weight coefficients.
We also prove an explicit separation between monotone Boolean
formulas and monotone real formulas. 
Previously the result was only known to hold
non-constructively.
To obtain these results we strengthen a lifting theorem of~\cite{PR18LiftingNS}
to allow the lifting to work with \emph{any} gadget with sufficiently large rank, 
in particular with the equality gadget---a crucial ingredient for 
obtaining the separations discussed above. 

This work raises a number of questions. 
Prior to our result,
no explicit function was known separating 
monotone real circuits or formulas from monotone Boolean circuits or formula.
Although we prove an explicit formula separation, it remains open
to obtain an explicit function that separates monotone real circuits
from monotone Boolean circuits.

The most glaring open problem related to our cutting planes contribution
is to strengthen our result to a true length separation, without any
assumption on the space complexity.
It is natural to ask whether techniques inspired 
by~\cite{Sokolov17Daglike,GGKS18MonotoneCircuit} can be of use.
Another thing to note about our trade-off result for \cpstar
is that it is not a ``true trade-off'': we know that length
and space cannot be optimised simultaneously, but we do not know 
if there in fact exist small space refutations. An interesting 
problem is, therefore, to exhibit formulas that present ``true trade-offs''
for \cpstar but are easy with regard to space and length in CP.

It follows from our results that standard decision tree complexity, parity decision tree complexity,
and Nullstellensatz degree are equal for the falsified clause search problem of lifted pebbling formulas.
In view of this we can ask ourselves what complexity measure we are actually lifting.
We know that for general search problem decision tree complexity is not enough for a lifting result.
How about parity decision tree complexity?
Or can we leverage the fact that we have ``well-behaved'' rectangle covers
and small certificate complexity to lift weaker complexity models?
It would be valuable to have a better understanding of the relation
between gadgets, outer functions/relations and complexity measures.

\section*{Acknowledgements}

Different subsets of the authors would like to acknowledge fruitful and
enlightening conversations with different subsets of
Arkadev Chattopadhyay,
Pavel Hrubeš,
Christian~Ikenmeyer,
Bruno Loff,
Sagnik Mukhopadhyay,
Igor Carboni Oliveira,
Pavel Pudlák,
and
Dmitry Sokolov.
We are also grateful for discussions regarding literature references with
Albert Atserias,
Paul Beame,
and
\mbox{Massimo Lauria.}
We are 
thankful to the
anonymous referees for their comments; in particular, indicating a simplified proof of Lemma \ref{lem:alternative-good}.

Part of this work was carried out while several of the authors were
visiting the Simons Institute for the Theory of Computing in
association with the DIMACS/Simons Collaboration on Lower Bounds in
Computational Complexity, which is conducted with support from the
National Science Foundation.

Or Meir was supported by the Israel Science Foundation (grant No.~1445/16).
Toniann Pitassi and Robert Robere were  supported by NSERC.
Susanna F.~de Rezende and Jakob Nordström were supported
by the European Research Council under the
European Union's Seventh Framework Programme \mbox{(FP7/2007--2013) /}
ERC grant agreement no.~279611,
as well as by the
Knut and Alice Wallenberg grant
KAW 2016.0066.
Jakob Nordström also received funding from the
Swedish \mbox{Research} Council grants
\mbox{621-2012-5645}
and
\mbox{2016-00782}.
Marc Vinyals was supported by
the Prof.\ R Narasimhan post-doctoral award.

\appendix

\section{Lifting Nullstellensatz Degree for All Gadgets}
\label{sec:nullstellensatz-lifting}

In this section we prove Theorem \ref{cor:nss-communication-eq}. In fact, we prove the following stronger result, which implies Theorem~\ref{cor:nss-communication-eq} as a corollary.

\begin{theorem}\label{th:ns-lifting}
Let $\mathcal{C}$ be an unsatisfiable $k$-CNF on $n$ variables and let $\mathbb{F}$ be any field.
Let $g$ be any Boolean-valued gadget with $\rank(g) \geq 4$.
Then \[ \mathsf{P}^{\cc}(\Search(\mathcal{C}) \circ g^n) \geq \NS_{\mathbb{F}}(\mathcal{C}) \log \left(\frac{\NS_{\mathbb{F}}(\mathcal{C})\rank(g)}{en}\right) - \frac{6n\log e}{\rank(g)} - \log k.\]
\end{theorem}

The proof of the theorem follows the proof of a similar lifting theorem from \cite{PR18LiftingNS}.
As such, we will need some notation from that paper.
Let us begin by introducing a key notion: Razborov's \emph{rank measure}.
Given sets $\mathcal{U}, \mathcal{V}$, a \emph{rectangle cover} $\mathcal{R}$ of $\mathcal{U} \times \mathcal{V}$ is a covering of $\mathcal{U} \times \mathcal{V}$ by combinatorial rectangles.

\begin{definition}\label{def:rank-measure}
  Let $\mathcal{U}, \mathcal{V}$ be sets and let $\mathcal{R}$ be any rectangle cover of $\mathcal{U} \times \mathcal{V}$.
  Let $A$ be any $\mathcal{U} \times \mathcal{V}$ matrix over a field $\mathbb{F}$.
  The \emph{rank measure} of $\mathcal{R}$ at $A$ is the quantity \[ \mu_{\mathbb{F}}(\mathcal{R}, A) = \frac{\rank(A)}{\displaystyle \max_{R \in \mathcal{R}} \rank(A \restriction R)}.\]
\end{definition}

Using the rank measure we can lower bound the deterministic communication complexity of composed CNF search problems as follows.
The key observation is that any deterministic communication protocol outputs a rectangles that lie in a ``structured'' rectangle cover in the following sense.
We note below that if $A$ is a collection of tuples from some product set $\mathcal{I}^n$ then we write $A_i$ to mean the projection of $A$ to the $i$th coordinate and $A_I$ for $I \subseteq [n]$ to mean the projection onto the coordinates in $I$.

\begin{definition}\label{def:structured-rectangle-cover}
  Let $\mathcal{C}$ be an unsatisfiable $k$-CNF on $n$ variables and let $g : \mathcal{X} \times \mathcal{Y} \rightarrow \set{0,1}$ be a gadget.
  For a clause $C \in \mathcal{C}$, a combinatorial rectangle $R \subseteq \mathcal{X}^n \times \mathcal{Y}^n$ is \emph{$C$-structured} if $g^n(x, y)$ falsifies $C$ for all $(x, y) \in R$ and for all $i \not \in \vars{C}$ we have $R_i = \mathcal{X} \times \mathcal{Y}$.
  A rectangle cover $\mathcal{R}$ of $\mathcal{X}^n \times \mathcal{Y}^n$ is \emph{$\mathcal{C}$-structured} if every $R \in \mathcal{R}$ is $C$-structured for some $C \in \mathcal{C}$.
\end{definition}

\begin{lemma}\label{lem:cc-rank-lb}
  Let $\mathcal{C}$ be an unsatisfiable $k$-CNF on $n$ variables and let $g: \mathcal{X} \times \mathcal{Y} \rightarrow \set{0,1}$ be a gadget.
  Let $\mathbb{F}$ be any field and let $A$ be any $\mathcal{X}^n \times \mathcal{Y}^n$ matrix over $\mathbb{F}$.
  Then \[\Pcc(\Search(\mathcal{C}) \circ g) \geq \min_{\mathcal{R}} \log \mu_\mathbb{F}(\mathcal{R}, A) \] where the minimum is taken over $\mathcal{C}$-structured rectangle covers of $\mathcal{X}^n \times \mathcal{Y}^n$.
\end{lemma}
\begin{proof}
  Let $\Pi$ be any communication protocol solving $\Search(\mathcal{C}) \circ g$, and let $\mathcal{T}$ be the monochromatic rectangle partition corresponding to $\Pi$.
  Since $\Pi$ solves the search problem, for every rectangle $R \in \mathcal{T}$ there is a clause $C$ such that for all $(x, y) \in R$, $g^n(x, y)$ falsifies $C$.
  We can write $R = A \times B$ for some sets $A \subseteq \mathcal{X}^{\vars{C}} \times \mathcal{X}^{[n] \setminus \vars{C}}$ and $B \subseteq \mathcal{Y}^{\vars{C}} \times \mathcal{Y}^{[n] \setminus \vars{C}}$.
  Consider $R' = A' \times B'$ where $A' = A_{\vars{C}} \times \mathcal{X}^{[n] \setminus \vars{C}}$ and $B' = B_{\vars{C}} \times \mathcal{X}^{[n] \setminus \vars{C}}$.
  Since $C$ only depends on indices in $\vars{C}$ we have that $g^n(x, y)$ falsifies $C$ for all $(x, y) \in R'$ and, moreover, $R'_i = \mathcal{X} \times \mathcal{Y}$ for all $i \not \in \Vars{C}$.
  It follows that $R'$ is $C$-structured. Let $\mathcal{R}$ be the $\mathcal{C}$-structured rectangle covering obtained from $\mathcal{T}$ by relaxing all rectangles of $\mathcal{T}$ in this way.

  We now have an $\mathcal{C}$-structured rectangle cover $\mathcal{R}$ such that every $T \in \mathcal{T}$ is contained in some rectangle of $\mathcal{R}$.
  Razborov \cite{Razborov90Applications} proved that if $\mathcal{T}$ is a rectangle partition and $\mathcal{R}$ is a rectangle cover such that for each $T \in \mathcal{T}$ there is an $R \in \mathcal{R}$ such that $T \subseteq R$ it holds that \[ |\mathcal{T}| \geq \mu_{\mathbb{F}}(\mathcal{R}, A)\] for any matrix $A$.
  Since $\log {|\mathcal{T}|} \leq |\Pi| = \Pcc(\Search(\mathcal{C}) \circ g)$ the lemma follows.
\end{proof}

We now introduce the notion of a \emph{certificate} of an unsatisfiable CNF formula.

\begin{definition}\label{def:cnf-certificates}
  Let $\mathcal{C}$ be an unsatisfiable Boolean formula on $n$ variables in conjunctive normal form, and let $C$ be a clause in $\mathcal{C}$.
  The \emph{certificate} of $C$, denoted $\Cert(C)$, is the partial assignment $\pi: [n] \rightarrow \set{0,1,*}$ which falsifies $C$ and sets the maximal number of variables to $*$s.
  Let $\Cert(\mathcal{C})$ denote the set of certificates of clauses of $\mathcal{C}$.
\end{definition}

We say that an assignment $z \in \set{0,1}^n$ \emph{agrees} with a certificate $\pi \in \Cert(\mathcal{C})$ if $\pi(i) = z_i$ for each $i$ assigned to a $\set{0,1}$ value by $\pi$.
Since the CNF formula $\mathcal{C}$ is unsatisfiable, it follows that every assignment in $z \in \set{0,1}^n$ agrees with some $\set{0,1}$-certificate of $\mathcal{C}$.
Next we introduce an alternative definition of Nullstellensatz degree called the \emph{algebraic gap complexity}.

\begin{definition}\label{def:algebraic-gaps-general}
  Let $\mathbb{F}$ be a field.
  Let $\mathcal{C}$ be an unsatisfiable CNF on $n$ variables.
  The \emph{$\mathbb{F}$-algebraic gap complexity} of $\mathcal{C}$ is the maximum positive integer $\gap_{\mathbb{F}}(\mathcal{C}) \in \mathbb{N}$ for which there exists a multilinear polynomial $p \in \mathbb{F}[z_1, z_2, \ldots, z_n]$ such that
  \[ \deg(p) = n \quad \text{ and} \quad \forall \pi \in \Cert(\mathcal{C}): \deg(p \restriction \pi) \leq n - \gap_{\mathbb{F}}(\mathcal{C}) \eqperiod\]
  When the field is clear from context we will write $\gap(\mathcal{C})$.
\end{definition}

In \cite{PR18LiftingNS, Robere18Thesis} it was shown that the algebraic gap complexity is equal to Nullstellensatz degree.

\begin{theorem}\label{th:gap-equals-ns}
  For any unsatisfiable CNF formula $\mathcal{C}$ on $n$ variables and any field $\mathbb{F}$, $\gap_{\mathbb{F}}(\mathcal{C}) = \NS_{\mathbb{F}}(\mathcal{C})$.
\end{theorem}

We now prove a lifting theorem from Nullstellensatz degree to the rank measure, from which Theorem \ref{th:ns-lifting} follows by applying Lemma \ref{lem:cc-rank-lb}.
\begin{theorem}\label{th:rank-measure-lifting}
  Let $\mathcal{C}$ be an unsatisfiable $k$-CNF on $n$ variables and let $\mathbb{F}$ be any field.
  Let $g$ be any Boolean-valued gadget with $\rank(g) \geq 4$.
  There is a matrix $A$ such that for any $\mathcal{C}$-structured rectangle cover $\mathcal{R}$ we have \[ \mu_{\mathbb{F}}(\mathcal{R}, A) \geq \frac{1}{k}\left(\frac{\NS_{\mathbb{F}}(\mathcal{C})\rank(g)}{en}\right)^{\NS_{\mathbb{F}}(\mathcal{C})} \exp(-6n/\rank(g))\eqperiod \]
\end{theorem}
\begin{proof}
  Let $p \in \mathbb{F}[z_1, z_2, \ldots, z_n]$ be the polynomial witnessing the algebraic gap complexity $\gap(\mathcal{C})$, and let $A = p \circ g^n$ be the pattern matrix obtained by composing $p$ and $g$.
  We need to analyze \[ \mu_{\mathbb{F}}(\mathcal{R}, p \circ g^n) = \frac{\rank_{\mathbb{F}}(p \circ g^n)}{\displaystyle{\max_{R \in \mathcal{R}}\rank_{\mathbb{F}}(p \circ g^n \restriction R)}}\eqperiod \]

  Let us first analyze the denominator.
  Let $R$ be an arbitrary rectangle from the cover $\mathcal{R}$, and suppose that $R$ is $C$-structured for the clause $C \in \mathcal{C}$.
  Let $\pi = \Cert(C)$.
  We want to show that
  \begin{equation}
    \label{eq:rank-restricted-new}
    \rank_{\mathbb{F}}(p \circ g^n \restriction R) \leq \sum_{S: \widehat{p \restriction \pi}(S) \neq 0} \rank(g)^{|S|} \eqperiod
  \end{equation}
  To prove this, we claim that $p \circ g^n \restriction R$ is column-equivalent to the block matrix \[ [(p \restriction \pi) \circ g^{[n] \setminus \vars{C}}, (p \restriction \pi) \circ g^{[n] \setminus \vars{C}}, \ldots, (p \restriction \pi) \circ g^{[n] \setminus \vars{C}}]\] for some number of copies of the matrix $(p \restriction \pi) \circ g^{[n] \setminus \vars{C}}$.
  Indeed Equation \ref{eq:rank-restricted-new} immediately follows from this claim as \[\rank_{\mathbb{F}}(p \circ g^n \restriction R) = \rank_{\mathbb{F}}((p \restriction \pi) \circ  g^{[n] \setminus \vars{C}}) \leq \sum_{S: \widehat{p \restriction \pi}(S) \neq 0} \rank(g)^{|S|}\] by Theorem \ref{th:rank-lifting}.
  So, we now prove the claim.

  Write $R = A \times B$.
  Fix assignments $\alpha \in A_{\vars{C}}$ and $\beta \in B_{\vars{C}}$, and note that since $R$ is $C$-structured we have that $g^{\vars{C}}(\alpha, \beta) = \pi$ and $(\alpha, x') \in A$ and $(\beta, y') \in B$ for all $x', y'$.
  Thus, by ranging $x_{[n] \setminus \vars{C}}, y_{[n] \setminus \vars{C}}$ over all values yields the matrix $(p \restriction \pi) \circ g^{[n] \setminus \vars{C}}$.
  Then, ranging $x_{\vars{C}}$ and $y_{\vars{C}}$ over all $\alpha, \beta$ such that $g^{\vars{C}}(\alpha, \beta) = \pi$ yields the claim and Equation \ref{eq:rank-restricted-new}.

  Now, consider the rank measure $\mu_{\mathbb{F}}(\mathcal{R})$, which by Theorem \ref{th:rank-lifting} and Equation \ref{eq:rank-restricted-new} satisfies \[ \mu_\mathbb{F}(\mathcal{R}) \geq \frac{\rank_{\mathbb{F}}(p \circ g^n)}{\displaystyle{\max_{R \in \mathcal{R}} \rank_{\mathbb{F}}(p \circ g^n \restriction R)}} = \frac{\displaystyle \sum_{S: \hat p(S) \neq 0} (\rank(g)-3)^{|S|}}{\displaystyle \max_{\pi \in \Cert(\mathcal{C})} \sum_{S : \widehat{p \restriction \pi}(S) \neq 0} \rank(g)^{|S|}}\]
  By definition of $\gap(\mathcal{C})$ we have $\deg p = n$ and thus the numerator is at least $(\rank(g) - 3)^n$.
  For the denominator, since $p$ witnesses the algebraic gap of $\mathcal{C}$, we have that $\deg p \restriction \pi \leq n - \gap(\mathcal{C})$ for all $\pi \in \Cert(\mathcal{C})$.
  We may assume that $\hat p(S) = 0$ when $|S| < n - \gap(\mathcal{C})$ as the definition of algebraic gaps depends only on the coefficients of monomials of $p$ with degree larger than $n - \gap(\mathcal{C})$. 
  So, for any restriction $\pi$:
\begin{align*}
    \sum_{S: \widehat{p \restriction \pi}(S) \neq 0} \rank(g)^{|S|} &\leq \sum_{i = 0}^k {n \choose \gap(\mathcal{C}) - i} \rank(g)^{n - \gap(\mathcal{C}) - i} \\
    & \leq k \left( \frac{en}{\gap(\mathcal{C})} \right)^{\gap(\mathcal{C})}\rank(g)^{n - \gap(\mathcal{C})}
  \end{align*}
  Putting it all together, and using the fact that $\rank(g) \geq 6en/\gap(\mathcal{C})$, we have
  \begin{align*}
    \mu_{\mathbb{F}}(\mathcal{R}, p \circ g^n) & \geq \frac{(\rank(g) - 3)^n}{k(en/\gap(\mathcal{C}))^{\gap(\mathcal{C})} \rank(g)^{n-\gap(\mathcal{C})}} \\
                                                          & = \frac{1}{k}\left(\frac{\gap(\mathcal{C})\rank(g)}{en}\right)^{\gap(\mathcal{C})} \cdot \left(1 - \frac
                                                            {3}{\rank(g)}\right)^n \\
                                                          & \geq \frac{1}{k}\left(\frac{\gap(\mathcal{C})\rank(g)}{en}\right)^{\gap(\mathcal{C})} \exp(-6n/\rank(g)) \eqperiod
  \end{align*}
  Since $\gap(\mathcal{C}) = \NS(\mathcal{C})$ the theorem is proved.
\end{proof}

Theorem \ref{th:ns-lifting} follows immediately from Theorem \ref{th:rank-measure-lifting} and Lemma \ref{lem:cc-rank-lb}.

\section{Proof of the Space Lemma}
\label{sec:space-lemma}

In this section we prove the Space Lemma
(\reflem{lem:clause-in-space-five}), restated next.
\begin{lemma}
  Let $\formf$ be a set of inequalities over $\numvariables$ variables
  that implies a clause $\clc$. Then there is a cutting planes
  derivation of $\clc$ from $\formf$ in length
  $\bigoh{\numvariables^2 2^\numvariables}$ and space $\bigoh{1}$.
\end{lemma}

We do so by adapting the proof in~\cite{GPT15SpaceComplexityCP} that
any formula has a cutting planes refutation in constant space in order
to show that, in fact, we can derive any clause that follows from a
set of inequalities in constant space.

At a bird's eye view, the proof in~\cite{GPT15SpaceComplexityCP} has
two steps. The primary step is building a refutation of the complete
tautology, the formula that contains all $2^n$ clauses with $n$
variables each forbiding one of the possible $2^n$ assignments, in
constant space. The authors come up with an order and a way to encode
that the first $K$ clauses are all true in small space for an
arbitrary $K$, and the rest of the primary step consists of showing
how to operate with this encoding also in small space, starting with
no clause being true and adding clauses one by one until a
contradiction arises. The secondary step is to transform the original
set of linear inequalities into the complete tautology.

If we do not start with an unsatisfiable set of linear inequalities we
obviously cannot reach a contradiction, but given a clause $\clc$ that
follows from $\formf$ we can still encode that all the clauses that
are a superset of $\clc$ must be true, and this expression is
equivalent to $\clc$.

Let us set up some notation. We number the variables from $0$ to
$\numvariables-1$. If $\alpha$ is a total assignment, we denote by
$\ctaxiom$ the clause over $n$ variables that is falsified exactly by
$\alpha$. We overload notation and also denote by $\ctaxiom$ the 
standard translation of the clause $\ctaxiom$ into an inequality.
We say that an assignment is less than a natural number
$\binarylimit$ and write $\alpha < \binarylimit$ if $\alpha$ is
lexicographically smaller than the binary representation of
$\binarylimit$, that is if
$\sum_{i=0}^{n-1}2^i \alpha(x_i) < \binarylimit$. We write
$\ctaccumulator$ to denote the inequality
$\sum_{i=0}^{n-1} 2^i x_i \geq \binarylimit$ that is falsified exactly
by the assignments
$\setdescr{\alpha\in\set{0,1}^n}{\alpha < \binarylimit}$.

We can reuse the following two intermediate lemmas
from~\cite{GPT15SpaceComplexityCP}, corresponding to the primary and
the secondary steps.

\begin{lemma}[\cite{GPT15SpaceComplexityCP}]
  \label{lem:ta-in-space-five}
 There is a cutting
  planes derivation of $\ctaccumulator$ from the set of clauses
  $\setdescr{\ctaxiom}{\alpha<\binarylimit}$ in length
  $\bigoh{\numvariables\binarylimit}$ and space $\bigoh{1}$.
\end{lemma}

\begin{lemma}[\cite{GPT15SpaceComplexityCP}]
  \label{lem:ia-in-space-four}
  If a total assignment $\alpha$ falsifies a set of inequalities $\formf$,
  then there is a cutting planes derivation of $\ctaxiom$ from
  $\formf$ in length $\bigoh{\numvariables}$ and space $\bigoh{1}$.
\end{lemma}

\reflem{lem:ta-in-space-five}, which contains the core of the
argument, follows from the proof of Lemma~3.2
in~\cite{GPT15SpaceComplexityCP}. We repeat Claim~3 in that proof,
which shows how to inductively derive $\ctaccumulator[\binarylimit'+1]$
from $\ctaccumulator[\binarylimit']$ and $\ctaxiom[\binarylimit']$, not $2^n$
but $\binarylimit$ times.

In turn \reflem{lem:ia-in-space-four} follows from the proof of Theorem~3.4
in~\cite{GPT15SpaceComplexityCP}: since $\alpha$ must falsify some
inequality $I$ from $\formf$, we only need to reproduce the derivation
of $\ctaxiom$ from $I$ verbatim.

\begin{proof}[Proof of \reflem{lem:clause-in-space-five}]
  Assume for now that
  $\clc = x_{n-1} \lor \cdots \lor x_{n-k}$. Consider the derivation
  $\Pi$ of the inequality $\ctaccumulator[2^{n-k}]$ from
  $\setdescr{\ctaxiom}{\alpha<\binarylimit}$, which is
  equivalent to $\clc$, given by \reflem{lem:ta-in-space-five}. We
  build a new derivation $\Pi'$ extending $\Pi$ as follows.

  Every time that we add an axiom $\ctaxiom$ to a configuration in
  $\Pi$, we replace that step by the derivation of $\ctaxiom$ from
  $\formf$ given by \reflem{lem:ia-in-space-four}. Observe that we
  only add axioms $\ctaxiom$ with $\alpha < 2^{n-k}$, and since any
  such assignment falsifies $\formf$ we meet the conditions to apply
  \reflem{lem:ia-in-space-four}.

  Finally we obtain $\clc$ from $\ctaccumulator[2^{n-k}]$ by
  considering $\set{\ctaccumulator[2^{n-k}]}$ as a set of inequalities
  over the $k$ variables $x_{n-1}\ldots x_{n-k}$ and applying
  \reflem{lem:ia-in-space-four} with $\alpha=0^k$ being the only
  assignment over these variables that falsifies
  $\ctaccumulator[2^{n-k}]$. The result is $\ctaxiom[0]=\clc$.

  To derive a general clause $\clc$ that contains $k'$ negative
  literals, say
  $\clc = \olnot{x_{n-1}} \lor \cdots \lor \olnot{x_{n-k'}} \lor
  x_{n-k'-1} \lor \cdots \lor x_{n-k}$, we build a derivation with the
  same structure as $\Pi'$, except that we replace every occurrence of $x_i$ by
  $(1-x_i)$ for $n-k' \leq i < n$. To do so, we replace each derivation of an axiom $\ctaxiom$ with a derivation
  of the axiom $\clc_{\alpha + (0^{n-k'}1^{k'})}$ and, for 
  $n-k' \leq i < n$, replace each use of
  $x_i \geq 0$ and $-x_i \geq -1$ by $-x_i \geq -1$ and $x_i \geq 0$,
  respectively. Linear combination and
  division steps go through unchanged, we only observe that at a
  division step the coefficient on the right hand side differs by a
  multiple of the divisor, so rounding is not affected.
\end{proof}

\section{Proof of Corollary~\ref{cor:query-complexity}}
\label{sec:query-complexity}

In this appendix, we provide the proof of  Corollary~\ref{cor:query-complexity}, restated next.

\begin{corollary}[\ref{cor:query-complexity}, restated]
For any field $\mathbb{F}$ and any directed acyclic graph $G$, the Nullstellensatz degree over $\mathbb{F}$ of $\pebblingformula$, the decision tree depth of $\Search(\pebblingformula)$, and the parity decision tree depth of $\Search(\pebblingformula)$ coincide and are equal to the reversible pebbling price of $G$.
\end{corollary}

\noindent Our proof uses Lemma~\ref{lem:pebbling-nss}, restated next.

\begin{lemma}[\ref{lem:pebbling-nss}, restated]
  For any field $\mathbb{F}$ and any graph $G$, $\NS_{\mathbb{F}}(\pebblingformula) = \rpeb(G)$.
\end{lemma}

Let $\mathbb{F}_2$ be the finite field of two elements. Given a search problem $\mathcal{S}$,  we denote the (deterministic) decision tree depth and parity decision tree depth of $\mathcal{S}$ by $\decisiontree(\mathcal{S})$ and $\paritydecisiontree(\mathcal{S})$ respectively. We use the following lemma, which will be proved below.

\begin{lemma}[Folklore]
  \label{lem:PDT-NS}
  Let $\formf$ be an unsatisfiable CNF over $n$ variables. Then, $\NS_{\mathbb{F}_2}(\formf) \le \paritydecisiontree(\Search(\formf))$.
\end{lemma}

\begin{proof}[Proof of Corollary~\ref{cor:query-complexity} from Lemmas \ref{lem:pebbling-nss} and \ref{lem:PDT-NS}]
Let $G$ be a directed acyclic graph.  Note that it suffices to prove the corollary for $\mathbb{F} = \mathbb{F}_{2}$, since Lemma~\ref{lem:pebbling-nss} implies that $\NS_{\mathbb{F}}(\pebblingformula)$ is the same for every finite field $\mathbb{F}$. The corollary follows immediately from the following chain of inequalities:
$$ \paritydecisiontree(\Search(\pebblingformula)) \le \decisiontree(\Search(\pebblingformula)) = \rpeb(G) =\NS_{\mathbb{F}_2}(\pebblingformula) \le \paritydecisiontree(\Search(\pebblingformula)). $$
The first inequality is obvious, and the first equality was proved in the work of Chan \cite{Chan13JustAPebble}, but for completeness we provide a simplified proof in Appendix~\ref{sec:reversible}. The second equality follows from Lemma~\ref{lem:pebbling-nss}, and the last inequality follows from Lemma~\ref{lem:PDT-NS}. Thus, the corollary is proved.
\end{proof}

In the remainder of this appendix, we prove Lemma~\ref{lem:PDT-NS}. Let $\formf$ be an unsatisfiable CNF over variables $z_1, \ldots, z_n$, and let $T$ be a parity decision tree of depth~$d$ that solves $\Search(\formf)$. We prove that there exists a Nullstellensatz refutation over $\mathbb{F}_2$ for $\formf$ of degree at most $d$, and this will imply the required result.

Recall that the parity decision tree $T$ takes as input an assignment $\alpha$ to $z_1, \ldots, z_n$, queries at most $d$ parities of $\alpha$, and then outputs a clause $C$ of $\formf$ that is violated by $\alpha$. More formally, every internal node $v$ of $T$ is associated with some linear polynomial $p_v$ in $z_1, \ldots, z_n$, and each outgoing edge $e$ of $v$ is associated with bit $b_e \in \{ 0,1 \}$ (so the edge $e$ is taken if $p_v(\alpha) = b_e$). Every leaf $\ell$ of $T$ is associated with a clause $C_\ell \in \formf$, such that every assignment $\alpha$ that leads $T$ to $\ell$ violates the clasue $C_\ell$.

We construct for each leaf $\ell$ of $T$ a polynomial $r_\ell(z_1, \ldots, z_n)$ of degree at most $d$ that output $1$ on an assignment $\alpha$ if the tree $T$ reaches the leaf $\ell$ when invoked on $\alpha$, and outputs $0$ otherwise. We will use those polynomials  later to construct the Nullstellensatz refutation of $\formf$. First, for every internal vertex $v$ and an outgoing edge $e$ of $v$, we associate with $e$ the linear polynomial
$$ r_e(z_1, \ldots, z_n) = p_v(z_1, \ldots, z_n) + b_e + 1.$$
Intuitively, the polynomial $r_e$ output $1$  on an assignment $\alpha$ if the query of $v$ outputs $b_e$ on $\alpha$, and $0$ otherwise. Now, to construct the polynomial $r_\ell$ of a leaf $\ell$, we multiply the polynomials $r_e$ for all the edges $e$ on the path from the root to $\ell$.  Since there are at most $d$ edges on that path, it follows that $r_\ell$ is of degree at most $d$. Moreover, it is not hard to see that $r_\ell (\alpha) = 1$ if $\alpha$ leads the tree $T$ to the leaf $\ell$, and $r_\ell (\alpha) = 0$ otherwise.

Let us denote by $r_{\ell}'$ the ``multilinearized" version of $r_\ell$, that is,  $r_{\ell}'$ is the polynomial obtained from $r_\ell$ by reducing the degree of every variable to $1$ in each of its occurences in $r_\ell$. It is not hard to see that $r_{\ell}'$ agrees with $r_\ell$ on every assignment in $\left\{0,1\right\}^n$.  Our Nullstellensatz refutation of $\formf$ is the polynomial
$$r = \sum_{\mathrm{leaf~}\ell\mathrm{~of~}T} r_{\ell}'.$$
Clearly, this polynomial is of degree at most $d$. In order to prove that the polynomial $r$ is a valid Nullstellensatz refutation of $\formf$, we need to prove that $r$ equals $1$, and that it can be derived from the axioms of $\formf$ (in other words, that $r$ belongs to ideal generated by $\mathcal{E}(\formf) \cup \left\{ z_i^2 - z_i \right\}_{i \in [n]}$). We start by proving that $r$ equals $1$.

\begin{claim}
$r = 1$.
\end{claim}

\begin{proof}
We use the well-known fact that a multilinear polynomial is determined by its values on $\left\{0,1\right\}^n$: one way to see it is to observe that the multilinear monomials are a basis of the space of functions from $\left\{0,1\right\}^n$ to $\left\{0,1\right\}$, and therefore every such function has a unique representation as a multilinear polynomial. Since $r$ is multilinear, it therefore suffices to prove that $r$ outputs $1$ on every assignment in $\left\{0,1\right\}^n$.

Let $\alpha \in \left\{0,1\right\}^n$ be an assignment to $z_1, \ldots, z_n$. Let $\ell$ be the leaf that $T$ reaches when invoked on $\alpha$. Then, by the construction of $r_{\ell}'$, it holds that $r_{\ell}'(\alpha) = 1$ and that $r_{\ell '}'(\alpha) = 0$ for every other leaf $\ell '$ of $T$. It follows that $r(\alpha) = 1$, as required.
\end{proof}

It remains to show that $r$ can be derived from the axioms of $\formf$. To this end, we will show that each of the polynomials $r_{\ell}'$ can be derived from those axioms . It is not hard to show that for every leaf $\ell$, the polynomial $r_\ell - r_{\ell}'$ can be derived from the boolean axioms $\left\{ z_i^2 - z_i \right\}_{i \in [n]}$ (in fact, this holds for any difference of a polynomial and its multilinearized version). Thus, it suffices to prove that for every leaf $\ell$, the polynomial $r_\ell$ can be derived from the axioms of $\formf$. We prove a stronger statement, namely, that for every leaf $\ell$, the polynomial $r_\ell$ is divisible by the polynomial $\mathcal{E}(C_\ell)$ (i.e., the polynomial encoding of the clause $C_\ell$). To this end, we prove the following result.

\begin{claim}
Let $p(z_1, \ldots, z_n)$ be a multilinear polynomial over $\mathbb{F}_2$, and let $i \in [n]$. If $p$ vanishes whenever $z_i = 0$, then $z_i$ divides $p$. Moreover, if $p$ vanishes whenever $z_i= 1$ , then $(1 - z_i)$ divides $p$.
\end{claim}

\begin{proof}
We can write $p = z_i \cdot a + b$, where $a$ and $b$ are polynomials that do not contain $z_i$. Suppose first that $p$ vanishes whenever $z_i = 0$. We would like to prove that $b = 0$. Assume that this is not the case. Then, there is an assignment $\alpha'$ to the variables $z_1, \ldots, z_n$ except for $z_i$ on which $b$  does not vanish. Now, if we extend $\alpha'$ to an assignment $\alpha$ to $z_1, \ldots, z_n$ by setting $z_i = 0$, it will follow that $\alpha$ is an assignment on which $z_i = 0$ but $p$ does not vanish, which is a contradiction.

Next, suppose that $p$ vanishes whenever $z_i = 1$. Observe that we can write $p = (1 -  z_i) \cdot a + (b - a)$ (here we use the fact that we are working over $\mathbb{F}_2$).  We would like to prove that $b - a = 0$. Assume that this is not the case. Then, there is an assignment $\alpha'$ to the variables $z_1, \ldots, z_n$ except for $z_i$ on which $b-a$  does not vanish. As before, if we extend $\alpha'$ to an assignment $\alpha$ to $z_1, \ldots, z_n$ by setting $z_i = 1$, it will follow that $\alpha$ is an assignment on which $z_i =1$ but $p$ does not vanish, which is a contradiction.
\end{proof}

We turn to prove that for every leaf $\ell$, the polynomial $r_\ell$ is divisible by the polynomial $\mathcal{E}(C_\ell)$. Fix a leaf $\ell$, and denote $C = C_\ell$. Recall that we denote by $C^+$ and $C^-$  the sets of variables that occur positively and negatively in $C$ respectively, and that
$$ \mathcal{E}(C) \equiv \prod_{z \in C^+}(1- z) \prod_{z \in C^-}z. $$
Now, observe that for every variable $z_i \in C^+$, it holds that $r_\ell$ vanishes whenever $z_i=1$: to see it, observe that when $z_i=1$, the assignment does not violate the clause $C$, and therefore the tree $T$ cannot reach $\ell$ when invoked on that assignment. Thus, $(1 - z_i)$ divides $r_\ell$ for every $z_i
\in C^+$.  Similarly, the variable $z_i$ divides $r_\ell$ for every $z_i \in C^-$. It follows that $r_\ell$ is divisible by every factor of $\mathcal{E}(C)$, and therefore it is divisible by $\mathcal{E}(C)$ (here we used the fact that each factor occurs in $\mathcal{E}(C)$ at most once). This concludes the proof.
 
\section{Reversible Pebbling is Equivalent to Query Complexity}
\label{sec:reversible}

In this appendix we present a direct proof that the reversible pebbling price of a graph
equals the query complexity of the search problem of the pebbling
formula of that graph, originally proved by
Chan~\cite{Chan13JustAPebble}.

\begin{theorem}[\cite{Chan13JustAPebble}]
  \label{th:justapebblegame}
  For every DAG $G$ with a single sink, it holds that $\decisiontree(\Search(\pebblingformula)) = \rpeb(G)$.
\end{theorem}

Let us introduce some notation to talk more formally about pebbling:
A \emph{pebbling configuration} is a set of vertices $P$. A  \emph{(reversible) pebbling} is a
sequence of configurations $\pebbling = P_1,\ldots,P_\ell$ where
$P_{i+1}$ follows from $P_i$ by applying the pebbling rules.
Its reverse $\reversepebbling(\pebbling)=P_\ell,\ldots,P_1$ is also a
valid pebbling. Its \emph{cost} is the maximal size of a configuration $P_i$.
Unless we call a pebbling \emph{partial}, we assume that
$P_1=\emptyset$.
A pebbling \emph{visits} $x$ if $x\in P_\ell$, and \emph{surrounds} $x$ if
$\pred(x) \subseteq P_\ell$.
The \emph{pebbling price} of a graph~$G$, denoted
$\rpeb(G)$, is the minimum cost of all pebblings that visit the sink,
and its \emph{surrounding price}, denoted $\speb(G)$, is the minimum cost of all
pebblings that surround the sink.

Given a decision tree for $\Search(\pebblingformula)$, we associate each node
with a state formed by a pair $(Q,Z)$ of queried vertices and the vertices
$Z \subseteq Q$ whose queries were answered by $0$.
It is immediate to verify that at a leaf either the sink $z$ belongs to $Q \setminus Z$,
or there is a vertex in $Z$ such that all of its
predecessors are in $Q \setminus Z$.
It is useful to generalize the definition of the search problem $\Search(\pebblingformula)$
to start with intermediate states. Specifically, we associate a state $(Q,Z)$ with the search problem
in which we are given an assignment to $\pebblingformula$ that is promised to assign $0$
to the vertices in $Z$ and $1$ to the vertices in~$Q \setminus Z$, and would like
to find a clause that is falsified by this assignment. We denote this search problem by
$\Search_G(Q,Z)$ and denote its query complexity by $\decisiontree_G(Q,Z)$.
We omit $G$ from the latter notation when it is clear from the context.
The crux of our proof is the following lemma, which implies Theorem~\ref{th:justapebblegame}.

\begin{lemma}\label{lem:RM-game}
For every DAG $G$ with a single sink $z$, it holds that $\decisiontree_G(\set{z},\set{z}) = \speb(G)$.
\end{lemma}

We claim that Lemma~\ref{lem:RM-game} implies Theorem~\ref{th:justapebblegame}.
To see why, let $G$ be a DAG with a single sink $z$, and let $G'$ be the DAG obtained
from $G$ by adding a new sink $z'$ and an edge from $z$ to $z'$. Then, it is not hard to
see that $\Search_G'(\set{z'}, \set{z'}) = \Search(\pebblingformula))$, and that
every pebbling that surrounds the sink of $G'$ is pebbling that visits the sink of~$G$,
and vice versa. Hence, it holds that
$$ \decisiontree(\Search(\pebblingformula)) =\decisiontree_{G'}(\set{z'}, \set{z'})
= \speb(G') = \rpeb(G),$$
where the second equality follows from Lemma~\ref{lem:RM-game}. In the rest of this
appendix, we focus on proving Lemma~\ref{lem:RM-game}. To this end, fix a DAG $G$
with a single sink $z$.

We first show that the states of an optimal decision tree for  $\Search(\set{z},\set{z})$
are of a special form, which we call "path-like".
Specifically, we say that a state $(Q,Z)$ is \emph{path-like} if there is a path $P$ ending at
the sink such that $P \intersection Q = Z$. Observe that in a path-like state
there is a unique such path of maximal length that starts in a vertex in $Z$.
We denote this path by $P_Z$, and denote its first vertex by $\head(Z)$.
We say that a vertex $v$ is \emph{relevant} to a path-like state $(Q,Z)$ if there is a path $P_v$ from $v$
to $\head(Z)$ such that $P_v \intersection Q = \set{\head(Z)}$.
Observe that starting from a path-like state and querying a relevant
vertex yields a path-like state where the new path is $P_Z \union P_v$
if the answer is $0$ and $P_Z$ if it is $1$.

We now show that if $(Q,Z)$ is a path-like state, then every optimal decision tree
for $\Search(Q,Z)$ only queries relevant vertices. Observe that this implies that all the nodes of the tree
correspond to path-like states. Moreover, this result holds in particular for $(\set{z},\set{z})$,
 since it is a path-like state.

\begin{lemma}
  \label{lem:path-like}
  If $(Q,Z)$ is a path-like state then every optimal decision
  tree for $\Search(Q,Z)$ only queries vertices relevant to $(Q,Z)$.
\end{lemma}

\begin{proof}
The proof is by induction on the query complexity $p$ of $\Search(Q,Z)$.
The base case $p=0$ holds vacuously: the optimal decision tree
does not make any queries.
Assume that $p>0$. Fix an optimal decision tree $T$ for $(Q,Z)$, and
let $v$ be the query made by the root of $T$. Observe that the children
of the root of $T$ correspond to the states $(Q \union{v}, Z \union{v})$
and $(Q \union{v}, Z)$. Let $T_0, T_1$ be the sub-trees rooted at 
those children respectively, and note that these trees are optimal for the latter states
 and that their depth is at most $p-1$.

Suppose first that $v$ is a relevant query to $(Q,Z)$. In this case,
the states of the children of the root  are path-like.
By the induction assumption, the trees $T_0, T_1$ only make queries
that are relevant to the states  $(Q \union{v}, Z \union{v})$
and $(Q \union{v}, Z)$ respectively.
Observe that any such query is necessarily relevant to $(Q,Z)$ as well.
Therefore, every query of $T$ is relevant to $(Q,Z)$.

Next, suppose that $v$ is irrelevant to $(Q,Z)$. We consider two
separate cases: the case where $v$ belongs to the path $P_Z$, and
the case where it does not belong to $P_Z$.

\paragraph*{$v$ belongs to $P_Z$.}
Assume that $v$ belongs to the path $P_Z$. In this case, the state
$(Q \union{v}, Z \union{v})$ is path-like, and therefore by the
induction assumption the tree $T_0$  only makes queries that are
relevant to that state. We claim
that $T_0$ solves $\Search(Q,Z)$, which contradicts the assumption
that $T$ is optimal for $\Search(Q,Z)$.

To see why, suppose for the sake of contradiction that $T_0$ fails to solve
$\Search(Q,Z)$ on some assignment $\alpha$. Observe that the only case where this
can happen  is when $\alpha(x_v) = 1$ but $T_0$ outputs that the violated clause is
$C_v = x_v \vee \bigvee_{u \in \pred(v)} \neg x_u$.
Let $u$ be the predecessor of $v$ that lies on $P_Z$
(such $u$ must exist since $v$ is queried by $T$ and hence cannot be equal to
$\head(Z)$). Then, $u$ cannot belong to $Z$ (or otherwise $T_0$ could not have 
 output $C_v$) and therefore it cannot belong to $Q$ (since $P_Z \cap Q = Z$).
This means that $T_0$ must have queried $u$ in order to output $C_v$. However,
 $u$ is irrelevant to $(Q \union{v}, Z \union{v})$, so we reached a contradiction
to the assumption that $T_0$ does not make irrelevant queries. Hence, $T_0$ solves
$\Search(Q,Z)$, as required.

\paragraph*{$v$ does not belong to $P_Z$.}
Next, assume that $v$ does not belong to the path $P_Z$. In this case,
the state $(Q \union{v}, Z)$ is path-like, and  therefore by the induction
assumption the tree $T_1$  only makes queries that are relevant to  that state.
We claim that $T_1$ solves $\Search(Q,Z)$,
which contradicts the assumption that $T$ is optimal for $(Q,Z)$. 

To see why, suppose for the sake of contradiction that $T_1$ fails to solve
$\Search(Q,Z)$ on some assignment $\alpha$. Observe that the only case where this
can happen is when $\alpha(x_v) = 0$ but for some successor $w$ of $v$ the tree~$T_1$
outputs that the violated clause is $C_w = x_w \vee \bigvee_{u \in \pred(w)} \neg x_u$.
The query $w$ cannot be relevant to $(Q,Z)$, since otherwise $v$ would have been
relevant for $(Q,Z)$. Since $w$ is irrelevant to $(Q,Z)$, it cannot be queried by $T_1$,
and therefore it must belong to $Z$ in order for $T_1$ to output $C_w$.
Moreover, $w$ cannot be equal to $\head(Z)$, since otherwise $v$ would have been
relevant for $(Q,Z)$. Thus, $w$ has a predecessor $u$ in $P_Z$.

The vertex $u$ cannot  belong to $Z$ (or otherwise $T_1$ could not have output $C_w$)
and therefore it cannot belong to $Q$ (since $P_Z \cap Q = Z$).
This means that $T_1$ must have queried $u$ in order to output $C_w$. However,
$u$ is irrelevant to $(Q \union{v}, Z)$, so we reached a contradiction to the
assumption that $T_1$ does not make irrelevant queries. Hence, $T_1$ solves
$\Search(Q,Z)$
\end{proof}

The following two propositions prove the two directions of Lemma~\ref{lem:RM-game}.
In the proof of the first proposition we use the following notion:
a \emph{pebbling assuming free pebbles on a set $S$} is a partial pebbling $\pebbling$
such that $P_1 \subseteq S$, and its \emph{cost} is the maximum size of
$P_i \setminus S$.

\begin{proposition}
  \label{pro:dt-to-pebbling}
  If $\decisiontree(\set{z}, \set{z}) \le p$, then $\speb(G) \le p$.
\end{proposition}

\begin{proof}
  We prove the following stronger claim: if for some path-like
  state $(Q,Z)$ there is an optimal decision tree $T$ of depth $p$,
  then there is a pebbling that surrounds $\head(Z)$  of
  cost $p$ assuming free pebbles on $Q \setminus Z$.
To see that this claim implies the proposition observe that $(\set{z},\set{z})$ is
  path-like. Therefore, if $\decisiontree(\set{z}, \set{z}) \le p$ then
  the claim implies that there is a pebbling that surrounds $z$
  of cost at most~$p$ without free pebbles.

  We prove the claim by induction on $p$. The base case is when
  $p=0$, so $T$ consists of a single leaf. This means that there must
  be some vertex in $Z$ that is surrounded by vertices in $Q\setminus Z$.
  This vertex must be $\head(Z)$, since any other vertex $v \in Z$ has
  a predecessor in $P_Z$, and this predecessor cannot belong to $Q \setminus Z$.
  Hence, $\set{Q \setminus Z}$ is a surrounding  pebbling of $\head(Z)$
  of cost $0$ assuming free pebbles on $Q \setminus Z$, as required.

  We proceed to the induction step. Suppose that $p>0$.
  Let $v$ be the query made at the root of $T$. Let $T_0, T_1$  be the
  subtrees rooted at the children of $v$  that corresponds to the states
  $(Q \union{v}, Z \union{v})$ and $(Q \union{v}, Z)$ respectively. By
  Lemma~$\ref{lem:path-like}$, the query $v$ is relevant to $(Q,Z)$,
  and therefore the latter states are path-like.

  Observe that $\head(Z \union{v}) = v$. Hence, by applying
  the induction assumption to $T_0$, we obtain a  pebbling $\pebbling_0$
  that surrounds $v$ of cost at most $p-1$  assuming free pebbles on
$$(Q \union{v}) \setminus (Z \union{v}) = Q \setminus Z.$$
Furthermore, by applying the induction assumption to $T_1$, we obtain
  a pebbling $\pebbling_1$ that surrounds $\head(Z)$ of cost at most~$p-1$
  assuming free pebbles on $(Q \union{v}) \setminus Z$.

  We now construct a pebbling that surrounds $\head(Z)$ assuming
  free pebbles on $Q \setminus Z$  as follows: we first follow $\pebbling_0$,
  thus reaching a configuration that surrounds $v$. Then, we place a pebble
  on $v$ (unless it is already pebbled). Next, we follow
  $\setdescr{P \union{v}}{P\in\reversepebbling(\pebbling_0)}$ to remove all the pebbles that were placed
  by $\pebbling_0$ except for~$(Q \union{v}) \setminus Z$. Finally, we follow
  $\pebbling_1$ and reach a configuration that surrounds $\head(Z)$. It is not hard
  to see that this pebbling has cost at most~$p$
  assuming free pebbles on $Q \setminus Z$, as required.
\end{proof}

In the proof of the second proposition, we use the following notion:
Given a pebbling $\pebbling$, we define $\static(\pebbling) = \intersection \pebbling$
to be the set of pebbles that are always present in $\pebbling$, and the \emph{non-static
cost} of $\pebbling$ be the maximum size of
$P \setminus \static(\pebbling)$ for $P \in \pebbling$. We also denote the vertices reachable from $v$, including $v$ itself, by $\desc(v)$.

\begin{proposition}
  \label{pro:pebbling-to-dt}
  If $\speb(G) \le p$ then $\decisiontree(\set{z}, \set{z}) \le p$.
\end{proposition}

\begin{proof}
  We prove the following stronger claim: if for some state $(Q,Z)$
  there is a partial pebbling $\pebbling$ that surrounds a vertex
  $w\in Z$ of non-static cost $p$ with
  $\static(\pebbling) \subseteq Q \setminus Z$, then there is a
  decision tree of depth $p$ that solves $\Search(Q,Z)$.

  To see that this claim implies the proposition, observe that
  a pebbling
  $\pebbling$ that surrounds $z$ of cost $p$ starting from $\emptyset$
  has $\static(\pebbling) = \emptyset$. Thus, the claim implies that
  if such a pebbling exists, it holds that
  $\decisiontree(\set{z}, \set{z}) \le p$.

  If $\static(\pebbling)$ surrounds $w$ then all the predecessors of
  $w$ are in $Q \setminus Z$. Therefore, $\Search(Q,Z)$ can be solved
  without making any queries: the decision tree can immediately output
  that the clause
  $C_{w} = x_{w} \vee \bigvee_{u \in \pred(w)} \neg x_u$ is violated.

  Otherwise we prove the claim by induction on $p$. The base case
  $p=0$ is a particular instance of $\static(\pebbling)$ surrounding
  $w$, hence we turn to proving the induction step and suppose that
  $p>0$. Having $\ell$ denote the length of $\pebbling$, let $v$ be
  the earliest (in time) vertex to be placed at some time $m>1$ and
  not removed until time $\ell$. Note that $v$ exists because $w$ is
  not surrounded in some configuration of $\pebbling$.

  Let $\pebbling' = P_m,\ldots,P_\ell$ be the subpebbling of
  $\pebbling$ from time $m$ to time $\ell$. Observe that
  $\static(\pebbling') = \static(\pebbling) \union\set{v}$ by construction, and thus the
  non-static cost of $\pebbling'$ is at most~$p-1$. By applying the
  induction assumption to $\pebbling'$, it follows that there exists a
  decision tree $T_1$ for $\Search(Q \union{v}, Z)$ of depth $p-1$.

  Next, observe that $\reversepebbling(\pebbling')$ is a partial
  pebbling that \emph{surrounds $v$} of non-static cost $p-1$ with
  $\static(\reversepebbling(\pebbling')) =
  \static(\pebbling)\union\set{v}$, therefore
  $\pebbling'' = \setdescr{P \setminus \desc(v)}{P\in\reversepebbling(\pebbling')}$ is a partial pebbling
  that surrounds $v$ of non-static cost $p-1$ with
  $\static(\pebbling'') \subseteq \static(\pebbling)$. Hence, by applying the
  induction assumption to $\pebbling''$, it follows that there exists
  a decision tree $T_0$ for $\Search(Q \union{v}, Z \union {v})$ of
  depth $p-1$.

  We now construct a decision tree $T$ for $\Search(Q,Z)$ as follows:
  The tree $T$ queries $v$. If the answer is $0$, the tree proceeds
  by invoking $T_0$, and otherwise it invokes $T_1$. It is not hard to
  see that $T$ has depth at most $p$ and that it solves $\Search(Q,Z)$,
  as required.
\end{proof}

It is worth mentioning that we can prove Theorem~\ref{th:justapebblegame}
directly, without going through Lemma~\ref{lem:RM-game}. This can be done
by splitting the proofs of each of the propositions into two cases: the case
where $Z = \emptyset$ and the case where $Z \neq \emptyset$. In the first
case, we need to work with a visiting pebbling of the sink rather than
a surrounding pebbling of $\head(Z)$. However, this makes the proof
more cumbersome.

\bibliography{refShort}

\bibliographystyle{alpha}

\end{document}